%% file: papersigplan.tex
\documentclass{sigplanconf}

\usepackage{balance}
\usepackage{alltt}
\usepackage{aliascnt}
\usepackage{amsmath,stmaryrd,natbib,listings,amsthm,amssymb,mathpartir,mnsymbol}
\usepackage[usenames,dvipsnames]{color}

\newcommand{\Scribtexttt}[1]{{\texttt{#1}}}
\newcommand{\SColorize}[2]{\color{#1}{#2}}
\newcommand{\inColor}[2]{{\Scribtexttt{\SColorize{#1}{#2}}}}
\definecolor{PaleBlue}{rgb}{0.90,0.90,1.0}
\newcommand{\rackett}[1]{\inColor{black}{#1}}
\input{local-macros}

\input{preamble}

\usepackage{ifpdf}
\ifpdf
 \usepackage[pdftex]{graphicx}  
 \usepackage[pdftex]{hyperref}
 \else
 \usepackage[dvips]{graphicx}  
 \usepackage[dvips]{hyperref}
 \fi

\newtheorem{theorem}{Theorem}  
 
\newaliascnt{lemma}{theorem}  
\newtheorem{lemma}[lemma]{Lemma}  
\aliascntresetthe{lemma}

\conferenceinfo{DLS~'14}{October 20 - 24 2014, Portland, OR, USA}
\copyrightyear{2014}
\copyrightdata{978-1-4503-3211-8/14/10}
\doi{2661088.2661098}

\title{Abstracting Abstract Control\iftr{\\ \large{(Extended)}}}
\authorinfo{J. Ian Johnson}{Northeastern University}{ianj@ccs.neu.edu}
\authorinfo{David Van Horn}{University of Maryland}{dvanhorn@cs.umd.edu}

\begin{document}
\maketitle


\begin{abstract}
\input{abstract}
\end{abstract}

\section{Introduction}
\input{introduction}
\section{Computation through the lens of static analysis}\label{sec:analysis} 
\input{analysis}
\section{The AAM methodology}
\label{sec:aam}
\input{aam}
\section{A refinement for exact stacks}\label{sec:pushdown}
\input{pushdown}
\section{Stack inspection and recursive metafunctions}\label{sec:inspection}
\input{inspection}
\section{Relaxing contexts for delimited continuations}\label{sec:delim}
\input{shiftreset}
\section{Short-circuiting via ``summarization''}\label{sec:memo}
\input{memo}
\section{Related Work}
\input{related}
\section{Conclusion}
As the programming world continues to embrace behavioral values like functions and continuations, it becomes more important to import the powerful techniques pioneered by the first-order analysis and model-checking communities.
It is our view that systematic approaches to analysis construction are pivotal to scaling them to production programming languages.
We showed how to systematically construct executable concrete semantics that point-wise abstract to pushdown analyses of higher-order languages.
We bypass the automata theoretic approach so that we are not chained to a pushdown automaton to model features such as first-class composable control operators.
The techniques employed for pushdown analysis generalize gracefully to apply to non-pushdown models and give better precision than regular methods.

\acks We thank the anonymous reviewers of DLS 2014 for their detailed
reviews, which helped to improve the presentation and technical
content of the paper.
This material is based on research sponsored by DARPA under the
Automated Program Analysis for Cybersecurity (FA8750-12-2-0106)
project. The U.S. Government is authorized to reproduce and distribute
reprints for Governmental purposes notwithstanding any copyright
notation thereon.

\balance
\bibliographystyle{plainnat}
\bibliography{bibliography}

\end{document}

%% file: local-macros.tex
\newcommand{\iftr}[1]{#1}
\newcommand{\ifntr}[1]{}

\newcommand{\eg}{\textit{e.g.}}

\newcommand{\zcfa}{$0$CFA}
\newcommand{\kcfa}{$k$CFA}

\newcommand{\nat}{{\mathbb N}}
\newcommand{\mexpr}{e}
\newcommand{\mexpri}[1]{e_{#1}}
\newcommand{\mvar}{x}

\newcommand{\mval}{v}
\newcommand{\mlam}{\ell}
\newcommand{\maval}{\hat{v}}

\newcommand{\maddr}{a}
\newcommand{\maddralt}{b}
\newcommand{\menv}{\rho}
\newcommand{\mstore}{\sigma}
\newcommand{\mastore}{{\hat\sigma}}
\newcommand{\mlkont}{\iota}
\newcommand{\mkont}{\kappa}

\newcommand{\mmkont}{C}
\newcommand{\mamkont}{{\hat{C}}}
\newcommand{\mstate}{\varsigma}
\newcommand{\mastate}{\hat\varsigma}

\newcommand{\mtime}{t}
\newcommand{\mtimealt}{u}

\newcommand{\mktab}{\Xi}
\newcommand{\mmktab}{\chi}
\newcommand{\mmemo}{M}

\newcommand{\mctx}{\tau}
\newcommand{\mmctx}{\gamma}
\newcommand{\mactx}{\hat{\mctx}}
\newcommand{\msctx}{\dot{\mctx}}

\newcommand{\mkframe}{\phi}
\newcommand{\makframe}{\hat\phi}

\newcommand{\mtrace}{\pi}

\newcommand{\CEK}{\mathit{CEK}}
\newcommand{\CESK}{\mathit{CESK}}
\newcommand{\CESKt}{\mathit{CESK_t}}
\newcommand{\CESKstart}{\mathit{CESK^*_t}}
\newcommand{\CESKKstart}{\mathit{CESK^*_t\mktab}}

\newcommand{\CESIKKstart}{\mathit{CESIK^*_t\mktab}}
\newcommand{\CESIKKstar}{\mathit{CESIK^*\mktab}}

\newcommand{\SR}{\mathit{SR}}

\newcommand{\SRSChKKt}{\mathit{SRS\mmktab_t}}

\newcommand{\Var}{\mathit{Var}}

\newcommand{\Addr}{\mathit{Addr}}
\newcommand{\Time}{\mathit{Time}}

\newcommand{\Expr}{\mathit{Expr}}
\newcommand{\Store}{\mathit{Store}}
\newcommand{\Storeable}{\mathit{Storeable}}
\newcommand{\KStore}{\mathit{KStore}}
\newcommand{\MKStore}{\mathit{CStore}}
\newcommand{\LKont}{\mathit{LKont}}
\newcommand{\Kont}{\mathit{Kont}}
\newcommand{\VKont}{\widetilde{\mathit{Kont}}}
\newcommand{\SKont}{\widehat{\mathit{Kont}}}
\newcommand{\MKont}{\mathit{MKont}}
\newcommand{\Env}{\mathit{Env}}

\newcommand{\State}{\mathit{State}}

\newcommand{\Value}{\mathit{Value}}
\newcommand{\Lam}{\mathit{Lam}}

\newcommand{\mperm}{P}
\newcommand{\mpermmap}{m}
\newcommand{\Grant}{\mathbf{Grant}}
\newcommand{\Deny}{\mathbf{Deny}}
\newcommand{\mgd}{\mathit{gd}}
\newcommand{\Permissions}{\mathit{Permissions}}
\newcommand{\GD}{\mathit{GD}}
\newcommand{\PermissionMap}{\mathit{PermissionMap}}

\newcommand{\makont}{{\hat{\mkont}}}
\newcommand{\mvkont}{\tilde{\mkont}}

\newcommand{\Relevant}{\mathit{Relevant}}
\newcommand{\ExactContext}{\mathit{ExactContext}}
\newcommand{\Context}{\mathit{Context}}
\newcommand{\MContext}{\mathit{MContext}}

\newcommand{\Frame}{\mathit{Frame}}

\newcommand{\MKTab}{\mathit{KClosure}}
\newcommand{\Memo}{\mathit{Memo}}

\newcommand{\pop}{\mathit{pop}}
\newcommand{\popaux}{\mathit{pop}^*}

\newcommand{\fv}{\mathit{fv}}
\newcommand{\alloc}{\mathit{alloc}}
\newcommand{\tick}{\mathit{tick}}

\newcommand{\fresh}{\mathit{fresh}}

\newcommand{\passp}{\mathit{pass?}}

\newcommand{\dom}{\text{dom}}

\newcommand{\startstate}{\mathit{base}}
\newcommand{\tails}[2]{\mathit{tails}_{#1}(#2)}

\newcommand{\reify}{\mathit{reify}}

\newcommand{\reachrestrict}{\mathit{unfold}}
\newcommand{\stepextend}{\mathit{unfold}_1}

\newcommand{\approximate}{{\mathbb A}}
\newcommand{\touches}{{\mathcal T}}

\newcommand{\reaches}{{\mathcal R}}
\newcommand{\live}{\mathit{live}}

\newcommand{\inv}{\mathit{inv}}

\newcommand{\invmktab}{\mathit{inv}_\mktab}

\newcommand{\unroll}[2]{\mathit{unroll}_{#1}(#2)}

\newcommand{\unrollC}[2]{\mathit{unrollC}_{#1}(#2)}
\newcommand{\hastail}{\mathit{ht}}

\newcommand{\replacetail}{\mathit{rt}}

\newcommand{\kontlive}{K{\mathcal L}{\mathcal L}}
\newcommand{\kontliveaux}{K{\mathcal L}{\mathcal L}^*}
\newcommand{\terminal}{\mathit{terminal}}
\newcommand{\terminalaux}{\mathit{terminal}^*}

\newcommand{\post}{\mathit{post}}

\newcommand{\svar}[2][{}]{#2^{#1}}
\newcommand{\sapp}[3][{}]{(#2\ #3)#1}
\newcommand{\slam}[3][{}]{\lambda#1 #2. #3}

\newcommand{\sreset}[1]{\texttt{reset}\ #1}
\newcommand{\sshift}[2]{\texttt{shift}\ #1. #2}
\newcommand{\sframe}[2]{\texttt{frame}\ #1\ #2}
\newcommand{\sgrant}[2]{\texttt{grant}\ #1\ #2}
\newcommand{\stest}[3]{\texttt{test}\ #1\ #2\ #3}

\newcommand{\vcomp}[1]{\mathbf{comp}(#1)}

\newcommand{\appl}[1]{\mathbf{appL}(#1)}
\newcommand{\appr}[1]{\mathbf{appR}(#1)}

\newcommand{\ev}[1]{\mathbf{ev}(#1)}
\newcommand{\co}[1]{\mathbf{co}(#1)}

\newcommand{\kmt}{\mathbf{mt}}

\newcommand{\mkapp}[2]{#1\circ #2}

\newcommand{\sa}[1]{\widehat{\mathit{#1}}}

\newcommand{\somestate}{\lozenge}
\newcommand{\nextstate}{\blacklozenge}
\newcommand{\someotherstate}{\square}
\newcommand{\nextotherstate}{\blacksquare}

\newcommand{\lfp}{\mathbf{lfp}}
\newcommand{\tpl}[1]{\langle #1 \rangle}

\newcommand{\set}[1]{\{ #1 \}}
\newcommand{\setbuild}[2]{\{ #1\ :\ #2\}}

\newcommand{\kcons}[2]{#1{\tt :} #2}
\newcommand{\kconsm}[3]{#1{\tt :}^{#2} #3}
\newcommand{\alt}{\mathrel{\mid}}

\newcommand{\extm}[3]{#1[#2 \mapsto #3]}
\newcommand{\joinm}[3]{#1\sqcup[#2 \mapsto #3]}

\DeclareMathOperator*{\finto}{\rightharpoonup}
\DeclareMathOperator*{\deceq}{\overset{?}{=}}
\DeclareMathOperator*{\decin}{\overset{?}{\in}}

\newcommand{\OK}{{\mathcal O}{\mathcal K}}

\newcommand{\append}[2]{#1 {\tt ++} #2}

\newcommand{\stepto}{\longmapsto}
\newcommand{\kindastepto}{\dashedrightarrow}

%% file: preamble.tex
\usepackage{centernot}
\newcommand{\SCodePreSkip}{\vskip\abovedisplayskip}
\newcommand{\SCodePostSkip}{\vskip\belowdisplayskip}
\newenvironment{SCodeFlow}{\SCodePreSkip\begin{list}{}{\topsep=0pt\partopsep=0pt%
\listparindent=0pt\itemindent=0pt\labelwidth=0pt\leftmargin=2ex\rightmargin=2ex%
\itemsep=0pt\parsep=0pt}\item}{\end{list}\SCodePostSkip}
\newenvironment{SingleColumn}{\begin{list}{}{\topsep=0pt\partopsep=0pt%
\listparindent=0pt\itemindent=0pt\labelwidth=0pt\leftmargin=0pt\rightmargin=0pt%
\itemsep=0pt\parsep=0pt}\item}{\end{list}}
\newenvironment{RktBlk}{}{}
\definecolor{IdentifierColor}{rgb}{0.15,0.15,0.50}
\definecolor{ParenColor}{rgb}{0.52,0.24,0.14}
\newcommand{\RktSym}[1]{\inColor{IdentifierColor}{#1}}
\newcommand{\RktPn}[1]{\inColor{ParenColor}{#1}}

\definecolor{ValueColor}{rgb}{0.13,0.55,0.13}
\newcommand{\RktVal}[1]{\inColor{ValueColor}{#1}}

\newcommand{\ifwcm}[1]{}

%% file: abstract.tex
The strength of a dynamic language is also its weakness: run-time
flexibility comes at the cost of compile-time predictability.
Many of the hallmarks of dynamic languages such as closures,
continuations, various forms of reflection, and a lack of static types
make many programmers rejoice, while compiler writers, tool
developers, and verification engineers lament.
The dynamism of these features simply confounds statically reasoning
about programs that use them.
Consequently, static analyses for dynamic languages are few, far
between, and seldom sound.

The ``abstracting abstract machines'' (AAM) approach to constructing
static analyses has recently been proposed as a method to ameliorate
the difficulty of designing analyses for such language features.
The approach, so called because it derives a function for the sound
and computable approximation of program behavior starting from the
abstract machine semantics of a language, provides a viable approach
to dynamic language analysis since all that is required is a machine
description of the interpreter.

The AAM recipe as originally described produces finite state
abstractions: the behavior of a program is approximated as a finite
state machine.
Such a model is inherently imprecise when it comes to reasoning about
the control stack of the interpreter: a finite state machine cannot
faithfully represent a stack.
Recent advances have shown that higher-order programs can be
approximated with pushdown systems.
However, such models, founded in automata theory, either breakdown or
require significant engineering in the face of dynamic language
features that inspect or modify the control stack.

In this paper, we tackle the problem of bringing pushdown flow
analysis to the domain of dynamic language features.  We revise the
abstracting abstract machines technique to target the stronger
computational model of pushdown systems.
In place of automata theory, we use only abstract machines and
memoization.
As case studies, we show the technique applies to a language with
closures, garbage collection, stack-inspection, and first-class
composable continuations.

%% file: introduction.tex
Good static analyses use a combination of abstraction techniques, economical data structures, and a lot of engineering~\citep{dvanhorn:CousotEtAl-TASE07,ianjohnson:DBLP:journals/ipl/YiCKK07}.
The cited exemplary works stand out from a vast amount of work attacking the problem of statically analyzing languages like C.
Dynamic languages do not yet have such gems.
The problem space is different, bigger, and full of new challenges.
The traditional technique of pushing abstract values around a graph to get an analysis will not work.
The first problem we must solve is, ``what graph?'' as control-flow is now part of the problem domain.
Second, features like stack inspection and first-class continuations are not easily shoe-horned into a CFG representation of a program's behavior.
%

%
Luckily, there is an alternative to the CFG approach to analysis construction that is based instead on abstract machines, which are one step away from interpreters; they are interderivable in several instances~\citep{dvanhorn:Danvy:DSc}.
This alternative, called abstracting abstract machines
(AAM)~\citep{dvanhorn:VanHorn2012Systematic}, is a simple
idea that is generally applicable to even the most dynamic of
languages, \eg{},
JavaScript~\citep{ianjohnson:DBLP:journals/corr/KashyapDKWGSWH14}.
A downside is that all effective instantiations of AAM are finite state approximations.
Finite state techniques cannot precisely predict where a method or function call will return.
Dynamic languages have more sources for imprecision than non-dynamic languages (\eg{}, reflection, computed fields, runtime linking, {\tt eval}) that all need proper treatment in the abstract.
If we can't have precision in the presence of statically unknowable behavior, we should at least be able to \emph{contain} it in the states it actually affects.
Imprecise control flow due to finite state abstractions is an unacceptable containment mechanism.
It opens the flood gate to imprecision flowing everywhere through analyses' predictions.
It is also a solvable problem.
We extend the AAM technique to computably handle infinite state spaces by adapting pushdown abstraction methods to abstract machines.
The unbounded stack of pushdown systems is the mechanism to precisely match calls and returns.
We demonstrate the essence of our pushdown analysis construction by
first applying the AAM technique to a call-by-value functional
language (\S\ref{sec:aam}) and then revising the derivation to
incorporate an exact representation of the control stack
(\S\ref{sec:pushdown}).  We then show how the approach scales to
stack-reflecting language features such as garbage collection and
stack inspection (\S\ref{sec:inspection}), and stack-reifying features
in the form of first-class delimited control operators (\S\ref{sec:delim}).
These case studies show that the approach is robust in the presence of
features that need to inspect or alter the run-time stack, which
previously have required significant technical
innovations~\cite{dvanhorn:Vardoulakis2011Pushdown,ianjohnson:DBLP:journals/jfp/JohnsonSEMH14}.

Our approach appeals to operational intuitions rather than automata theory to justify the approach.
The intention is that the only prerequisite to designing a pushdown
analysis for a dynamic language is some experience with implementing
interpreters. 
%
%

%% file: analysis.tex
Static analysis is the process of soundly predicting properties of
programs.
It necessarily involves a tradeoff between the precision of those
predictions and the computational complexity of producing them.
At one end of the spectrum, an analysis may predict nothing, using no
resources.  At the other end, an analysis may predict everything, at
the cost of computability.

Abstract interpretation~\cite{dvanhorn:Cousot:1977:AI} is a form of
static analysis that involves the \emph{approximate} running of a
program by interpreting a program over an abstraction of the program's
values, e.g. by using intervals in place of
integers~\cite{Cousot-TASE07tutorial}, or types instead of
values~\cite{dvanhorn:esop:kmf07}.
By considering the sound abstract interpretation of a program, it is
possible to predict the behavior of concretely running the program. 
For example, if abstract running a program never causes a
buffer-overflow, run-time type error, or null-pointer dereference, we
can conclude actually running the program can never cause any of these
errors either.  If a fragment of code is not executed during the
abstract running, it can safely be deemed dead-code and removed.  More
fine-grained properties can be predicted too; to enable inlining, the
abstract running of a program can identify all of the functions that
are called exactly once and the corresponding call-site.  Temporal
properties can be discovered as well: perhaps we want to determine if
one function is always called before another, or if reads from a file
occur within the opening and closing of it.

In general, we can model the abstract running of a program by
considering each program state as a node in a graph, and track
evolution steps as edges, where each node and path through the graph
is an \emph{approximation} of concrete program behavior.
The art and science of static analysis design is the way we represent this graph of states; how little or how much detail we choose to represent in each state determines the precision and, often, the \emph{cost} of such an analysis.
First-order data-structures, numbers, arrays all have an abundance of
literature for precise and effective approximations, so this paper
focuses on higher-order data: closures and continuations, and their
interaction with state evolution.

A major issue with designing a higher-order abstract interpreter is
approximating closures and continuations in such a way that the
interpreter always terminates while still producing sound and precise
approximations.  Traditionally, both have been approximated by finite
sets, but in the case of continuations, this means the control stack
of the abstract interpreter is modeled as a finite graph and
therefore cannot be precise with regards to function calls and
returns.

\paragraph{Why pushdown return flow matters: an example}
Higher-order programs often create proxies, or monitors, to ensure an object or function interacts with another object or function in a sanitized way.
One example of this is behavioral contracts~\citep{dvanhorn:Findler2002Contracts}.
Simplified, here is how one might write an ad-hoc contract monitor for
a given function and predicates for its inputs and outputs:
 \begin{center}
\ifpdf
  \includegraphics[scale=0.45]{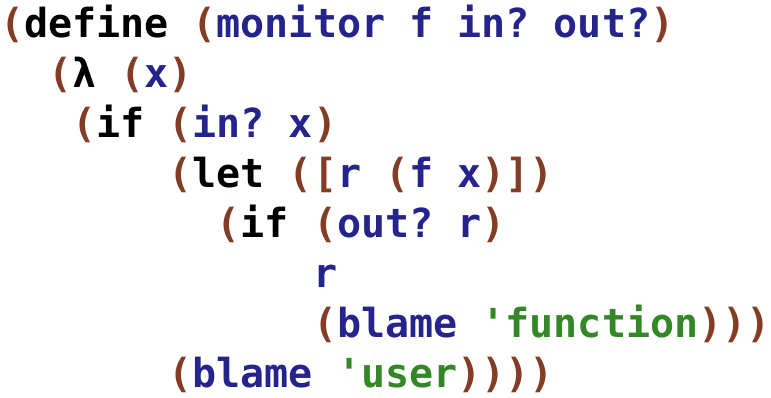}
\else
  \includegraphics[scale=0.45]{monitor.eps}
\fi
 \end{center}

It is well known that wrapping functions like this thwarts the
precision of regular \zcfa{} and higher \kcfa{} as more wrappings are
introduced.
In the case of this innocent program
 \begin{center}
\ifpdf
  \includegraphics[scale=0.45]{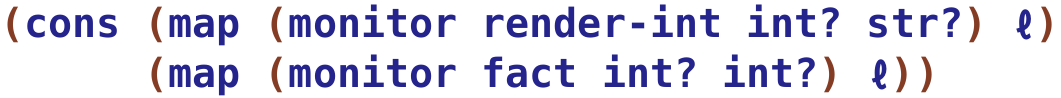}
\else
  \includegraphics[scale=0.45]{pair.eps}
\fi
 \end{center}

according to \zcfa{} the call to the wrapped \texttt{factorial}
function within the second \texttt{map} may return to within the
first \texttt{map}.  Hence \zcfa{} is not sufficiently precise to 
prove \texttt{factorial} cannot be blamed.
Using more a context-sensitive analysis such as 1CFA, 2CFA, etc.,
would solve the problem for this example, but would fail for nested
proxies.
In general, for any $k$, \kcfa{} will confuse the return flow of some
programs as in this example.
Yet, a pushdown abstraction that properly matches calls and returns
has no trouble with this example, regardless of proxy-nesting depth.

%

\paragraph{A systematic approach to pushdown analysis}

At this point, several pushdown analyses for higher-order languages
have been developed~\cite{dvanhorn:Vardoulakis2011CFA2,
dvanhorn:Earl2010Pushdown}, and the basic idea is simple: instead of
approximating a program with a finite state machine, use a pushdown
automata.  The control stack of the automata models the control stack
of the concrete interpreter, while stack frames, which contain
closures, are subject to the same abstraction as values in the
program.

This approach works well for simple languages which obey the stack
discipline of a PDA.  But most languages provide features that
transgress that discipline, such as garbage collection, first-class
control operators, stack inspection, and so on.  Some of these
features have been successfully combined with pushdown analysis, but
required technical innovation and
effort~\cite{dvanhorn:Vardoulakis2011Pushdown,
ianjohnson:DBLP:journals/jfp/JohnsonSEMH14,
dvanhorn:Earl2012Introspective}.  To avoid further one-off efforts, we
develop a general technique for creating pushdown analyses for
languages with control operators and reflective mechanisms.

%% file: aam.tex
Abstract machines are a versatile, clear and concise way to describe the semantics of programming languages.
They hit a sweet spot in terms of level of discourse, and viable implementation strategy.
First year graduate students learn programming language semantics with abstract machines~\cite{dvanhorn:Felleisen2009Semantics}.
%
They also turn out to be fairly simple to repurpose into static analyses for the languages they implement, via the Abstracting Abstract Machines methodology~\citep{dvanhorn:VanHorn2010Abstracting}.
The basic idea is that abstract machines implement a language's \emph{concrete} semantics, so we transform them slightly so that they also implement a language's \emph{abstract} semantics (thus ``abstracting'' abstract machines).

AAM is founded on three ideas:
\begin{enumerate}
\item{concrete and abstract semantics ideally should use the same code, for correctness and testing purposes,}
\item{the level of approximation should be a tunable parameter,}
\item{both of the above are achievable with a slight change to the abstract machine's state representation.}
\end{enumerate}

The first two points are the philosophy of AAM: correctness through simplicity, reusability, and sanity checking with concrete semantics.
The final point is the machinery that we recount in this section.

\subsection{The $\boldsymbol\CESKstart$ machine \emph{schema}}

The case studies in this paper have full implementations in PLT Redex~\citep{dvanhorn:Felleisen2009Semantics} available online.\footnote{\url{http://github.com/dvanhorn/aac}}
They all build off the untyped call-by-value lambda calculus, whose semantics we recall in small-step reduction semantics style.

\paragraph{From $\boldsymbol{\beta_v}$ to CEK - Being explicit:}

Expressions consist of variables, applications, and
$\lambda$-abstractions, which are considered values:
\[
\begin{array}{lclcl}
\mexpr & \in & \mathit{Expr} & ::= & \svar\mvar \alt \sapp{\mexpr}{\mexpr} \alt \mval\\[1mm]
\mval  & \in & \Value        & ::= & \slam{\mvar}{\mexpr}
\end{array}
\]
Reduction is characterized by a binary relation, $\beta_v$, which
reduces the application of a function to a value by substituting the
argument for the formal parameter of the function:
\[
\begin{array}{rcl}
(\slam{\mvar}{\mexpr})\;\mval & \beta_v & [\mval/\mvar]\mexpr
\end{array}
\]

A search strategy is given for specifying the order in which
function applications are selected for reduction.  This strategy is
specified as a grammar of evaluation contexts:
\[
\begin{array}{lclcl}
E & \in & \mathit{EvaluationContext} & ::= & [] \alt \sapp{E}{\mexpr} \alt \sapp{\mval}{E} 
\end{array}
\]
This grammar specifies that applications are reduced left-to-right,
and reduction does not occur under $\lambda$-abstractions.  The ``[]''
context is the empty ``hole'' context.  The notation $E[\mexpr]$ means
the expression obtained by replacing the hole of $E$ with
$\mexpr$.

\newcommand{\betastep}{\mathrel{\longmapsto_{\beta_v}}}
Finally, reduction on closed expressions is then given by the $\betastep$ relation:
\begin{align*}
\mexpr \betastep \mexpr' \text{ if } \mexpr = E[\mexpr_0], \mexpr' = E[\mexpr_1]\text{, and }\mexpr_0\;\beta_v\;\mexpr_1\text,\\
\text{
for some $E$, $\mexpr_0$, and $\mexpr_1$.}
\end{align*}

Although concise and well-structured for mathematical proofs, this
semantics sheds little light on the mechanics of an effective
implementation. There are essentially two aspects that need to be
modeled more concretely.

The first aspect is that of substitution.  The $\beta_v$ axiom models
function application via the meta-theoretic notion of substition, but
substitution can be modeled more explicitly by incorporating
an \emph{environment}, which maps a variable to the value that should
be substituted for it.  This gives rise to the notion of
a \emph{closure}, which represents values by $\lambda$-abstractions
paired together with an environment.  The $\beta_v$-axiom is replaced
by an alternative axiom that interprets function application by
extending the environment of a closure with the formal parameter
mapping to the argument in the scope of the body of the function;
it interprets variables, by looking up their value in the environment.

The second aspect is that of search.  The $\betastep$ relation is
defined in terms of evaluation contexts, but gives no guidance on how
to construct this context.  The search for reduction opportunities can
be modeled more explicitly by a stack that represents the evaluation
context in an ``inside-out'' fashion: the inner-most part of the
context is the top frame of the stack.

Modeling substution with closures and evaluation contexts with stacks
yields an abstract machine for the language, the CEK
machine~\cite{dvanhorn:Felleisen2009Semantics}, shown in \autoref{fig:cek}.

\begin{figure}\centering
  \begin{tabular}{rlrl}
    $\mstate \in \CEK$ &\hspace{-3mm}$::= \tpl{\mexpr, \menv, \mkont}$
    & $\mlam \in \Lam$ &\hspace{-3mm}$::= \slam{\mvar}{\mexpr}$\\
    $\mval \in \Value$ &\hspace{-3mm}$::= (\mlam,\menv)$ 
    & $\menv \in \Env$ &\hspace{-3mm}$= \Var \finto \Value$ \\
    $\mkont \in \Kont$ &\hspace{-3mm}$= \Frame^*$ &
    $\mvar \in \Var$ &\\ 
    $\mkframe \in \Frame$ &\multicolumn{2}{l}{\hspace{-3mm}$::= \appl{\mexpr,\menv} \alt \appr{\mval}$} \\[2mm]
  \end{tabular}

  $\mstate \stepto \mstate'$ \\[.5mm]
  \begin{tabular}{r|l}
    \hline\vspace{-3mm}\\
    $\tpl{\svar\mvar, \menv, \mkont}$
    &
    $\tpl{\mval, \mkont}$ if $\mval = \menv(\mvar)$
    \\
    $\tpl{\sapp{\mexpri0}{\mexpri1},\menv,\mkont}$
    &
    $\tpl{\mexpri0,\menv,\kcons{\appl{\mexpri1,\menv}}{\mkont}}$
    \\
    $\tpl{\mval, \kcons{\appl{\mexpr,\menv}}{\mkont}}$
    &
    $\tpl{\mexpr,\menv,\kcons{\appr{\mval}}{\mkont}}$
    \\
    $\tpl{\mval,\kcons{\appr{\slam{\mvar}{\mexpr},\menv}}{\mkont}}$
    &
    $\tpl{\mexpr,\extm{\menv}{\mvar}{\mval},\mkont}$
  \end{tabular}
  \caption{CEK machine}
  \label{fig:cek}
\end{figure}

\paragraph{From CEK to CESK$_t^\star$ - Accounting for space \& time:}

The CEK machine makes explicit the mechanics of substitution and
reduction in context, but doesn't give an explicit account of memory
allocation.  We now reformulate the CEK machine into a heap-based
machine that explicitly allocates memory for variable bindings and
stack frames.  Properly accounting for memory allocation will prove
essential for making computable abstractions, as the key mechanism for
soundly approximating the running of a program will be limiting the
memory available to analyze a program.

The memory heap is modeled as a \emph{store} mapping addresses to
values or frames.  Environments are modified to map variables to
addresses and when a function is applied, we allocate space in the
heap to hold the value of the argument.
To account for the space used by the stack, we change the
representation of the continuation from a inductively defined list of
frames to a linked list structure that is explicitly allocated in the
heap.

In anticipation of approximating machines, we make a slight
departure from the standard model and have the store map addresses to
\emph{sets of} values or frames, denoting the possible values an address may take on.

The CESK$_t^\star$ machine is defined in \autoref{fig:ceskstart}.
%
%
Typically, memory allocation is specified as $\maddr \notin \dom(\mstore)$, but this is just a specification, not an implementation.
Concretely, we rely on an allocation function $\alloc : \CESKstart \to \Addr$ that will produce an address for the machine to use.
The $\alloc$ function should be considered a parameter of the machine; we 
obtain different machines by instantiating different allocation strategies.
The allocation function need not always produce a fresh address (and finite
allocation strategies generate finite approximations of the running of a program).
It is important to observe that if $\alloc$ produces a previously
allocated address, the previous tenant of the address is not
forgotten, just merged as a \emph{possible} result.  This is reflected
in the use of $\sqcup$ when the store is updated, defined as:
\begin{align*}
  \joinm{\mstore}{\maddr}{s} &= \extm{\mstore}{\maddr}{\mstore(\maddr)\cup\set{s}}\text.
\end{align*}
Conversely, when an address is dereferenced, as in the case of
evaluating a variable or popping the top frame of the stack, a single
element is non-deterministically selected among those in the set for
that address.  If $\alloc$ always produces fresh addresses, this set
machinery degenerates so that singleton sets reside at every allocated
address and the non-determinism of address dereference vanishes.

In addition to the explicit machine management of memory, machine
states carry timestamps, written with subscripts.  Timestamps are
drawn from an unspecified set and rely on another parameter of the
machine: the $\tick$ function, which will play a role in the
approximating semantics, but for now can be thought of as a simple
counter.


%

%
%

\begin{figure}\centering
  \begin{tabular}{rlrl}
    $\mstate \in \CESKstart$ &\hspace{-3mm}$::= \tpl{\mexpr, \menv, \mstore, \mkont}_t$  &
$\mkont \in \Kont$ &\hspace{-3mm}$::= \epsilon \alt \kcons{\mkframe}{\maddr}$\\
 $\mstore \in \Store$ &\hspace{-3mm}$= \Addr \finto \wp(\Storeable)$ &
 $\maddr,\maddralt$ &\hspace{-3mm}$\in \Addr$\\ 
    $\menv \in \Env$ &\hspace{-3mm}$= \Var \finto \Addr$ 
& $\mtime,\mtimealt$ &\hspace{-3mm}$\in \Time$\\
    $\alloc$ &\hspace{-3mm}$: \CESKstart \to \Addr$ &
$\Storeable$ &\hspace{-3mm}$::= \mkont \alt \mval$
  \end{tabular}
\\[2mm]

  $\mstate \stepto \mstate'$ \quad $\maddr = \alloc(\mstate)$ \quad $\mtimealt = \tick(\mstate)$ \\
  \begin{tabular}{r|l}
    \hline\vspace{-3mm}\\
    $\tpl{\svar\mvar, \menv, \mstore,\mkont}_\mtime$
    &
    $\tpl{\mval, \mstore,\mkont}_\mtimealt$ if $\mval \in \mstore(\menv(\mvar))$
    \\
    $\tpl{\sapp{\mexpri0}{\mexpri1},\menv,\mstore,\mkont}_\mtime$
    &
    $\tpl{\mexpri0,\menv,\mstore',\kcons{\appl{\mexpri1,\menv}}{\maddr}}_\mtimealt$ \\
    where & $\mstore' = \joinm{\mstore}{\maddr}{\mkont}$
    \\
    $\tpl{\mval,\menv, \mstore,\kcons{\appl{\mexpr,\menv'}}{\maddralt}}_\mtime$
    &
    $\tpl{\mexpr,\menv',\mstore,\kcons{\appr{\mval,\menv}}{\maddralt}}_\mtimealt$
    \\
    $\tpl{\mval,\mstore,\kcons{\appr{\slam{\mvar}{\mexpr},\menv}}{\maddralt}}_\mtime$
    &
    $\tpl{\mexpr,\menv',\mstore',\mkont}_\mtimealt$ if $\mkont \in \mstore(\maddralt)$ \\
    where & $\menv' = \extm{\menv}{\mvar}{\maddr}$ \\
          & $\mstore' = \joinm{\mstore}{\maddr}{\mval}$
  \end{tabular} \\
  \caption{$\CESKstart$ semantics}
  \label{fig:ceskstart}
\end{figure}


\paragraph{Taming computability via \(\boldsymbol\alloc\):}
Allocation is key for both precision and abstraction power.
Fresh allocation concretely runs a program, precisely predicting everything; of course, analysis is undecidable.
If the allocator has a finite codomain, and there are no recursive datatypes in a state's representation, then the state space is finite -- the machine is a finite state machine.
If the allocator freshly allocates addresses for the stack representation, but is finite for everything else, the semantics is indistinguishable from just using the recursive representation of the stack.
If we have just the stack as an unbounded data structure, then the machine has finitely many stack frames and finitely many other components of the state -- the machine is a pushdown system.

In order to make the iteration of $\betastep$ on a program finite and
therefore have a computable graph, we must finitize the recursive
spaces in which the machine \emph{creates new values}.
AAM dictates that finitization can be centralized to one place,
address allocation, by redirecting values in recursive positions
through the store as we've already done.
%
%
To see why the reachable states of a program are finite, let's
consider each component of a machine state.  First, there is the
expression component.  For a given program, there a finite set of
subexpressions and no machine transition produces new expressions, so
there can only be a finite set of reachable expression components.
The environment component maps variable names to addresses.  Variables
are a subset of expressions, so there are a finite set of variables
(for a given program).  Since we've assumed $\alloc$ produces
addresses from a finite set, there can only be a finite set of
environments (the domain and range are finite sets).  The store
compoent maps addresses (finite) to values or continuations.  Values
are pairs of lambda-terms (a subset of expressions, hence finite) and
environments (as we've just determined: finite as well), so there are
a finite number of values.  Continuations are pairs of frames and
addresses.  Frames consist of expression-environment pairs or values,
both of which are finite sets, so continuations are finite too.
Therefore heaps are finite.  The last component is a continuation,
which we've just argued is finite.  Therefore, starting from an
initial configuration there can only be a finite number of states that
can be reached by the $\betastep$ relation assuming a finite $\alloc$
function.

If we run the $\CESKstart$ semantics to explore all possible states, we get a sound approximation of all paths that the $\CESK$ machine will explore, \emph{regardless of the particular $\alloc$ function used}; any $\alloc$ function is sound.
This paper will give a more focused view of the $\Kont$ component.
We said that when just $\Kont$ is unbounded, we have a pushdown system.
Pushdown systems cannot be naively \emph{run} to find all states and describe all paths the $\CESK$ machine can explore; the state space is infinite, therefore this strategy may not terminate.
The pushdown limitation is special because we can always recognize non-termination, stop, and describe the entire state space.
We show that a simple change in state representation can provide this functionality.
We regain the ability to just \emph{run} the semantics and get a finite object that describes all possible paths in the $\CESK$ machine, but with better precision than before.

%% file: pushdown.tex
We can exactly represent the stack in the $\CESKstart$ machine with a modified allocation scheme for stacks.
The key idea is that if the address is ``precise enough,'' then every path that leads to the allocation will proceed exactly the same way until the address is dereferenced.
\paragraph{``Precise enough'':}
For the $\CESKstart$ machine, every function evaluates the same way, regardless of the stack.
We should then represent the stack addresses as the components of a function call.
The one place in the $\CESKstart$ machine that continuations are allocated is at $\sapp{\mexpri0}{\mexpri1}$ evaluation.
The expression itself, the environment, the store and the timestamp are necessary components for evaluating $\sapp{\mexpri0}{\mexpri1}$, so then we just represent the stack address as those four things.
The stack is not relevant for its evaluation, so we do not want to store the stack addresses in the same store -- that would also lead to a recursive heap structure.
We will call this new table $\mktab$, because it looks like a stack.
By not storing the continuations in the value store, we separate ``relevant'' components from ``irrelevant'' components.
We split the stack store from the value store and use only the value store in stack addresses.
Stack addresses generally describe the relevant context that lead to their allocation, so we will refer to them henceforth as \emph{contexts}.
The resulting state space is updated here:
  \begin{align*}
    \sa{State} &= \sa{CESK}_t \times \KStore \\
    \Storeable &= \Value \\
    \mkont \in \Kont &::= \epsilon \alt \kcons{\mkframe}{\mctx} \\
    \mctx \in \Context &::=  \tpl{\mexpr,\menv,\mstore}_\mtime \\
    \mktab \in \KStore &= \Context \finto \wp(\Kont) \\
  \end{align*}

The semantics is modified slightly in \autoref{fig:ceskkstart-semantics} to use $\mktab$ instead of $\mstore$ for continuation allocation and lookup.
Given finite allocation, contexts are drawn from a finite space, but are still precise enough to describe an unbounded stack: they hold all the relevant components to find which stacks are possible.
The computed $\stepto$ relation thus represents the full description of a pushdown system of reachable states (and the set of paths).
Of course this semantics does not always define a pushdown system since $\alloc$ can have an unbounded codomain.
The correctness claim\ifntr{\footnote{All theorems and proofs are deferred to the extended version available on \url{arXiv.org}.}} is therefore a correspondence between the same machine but with an unbounded stack, no $\mktab$, and $\alloc, \tick$ functions that behave the same disregarding the different representations (a reasonable assumption).
\ifntr{See the extended paper for details.}

\begin{figure}
  \centering
  $\mastate,\mktab \stepto \mastate',\mktab'$ \quad $\maddr = \alloc(\mastate,\mktab)$ \quad $\mtimealt = \tick(\mastate,\mktab)$ \\
  \begin{tabular}{r|l}
    \hline\vspace{-3mm}\\
    $\tpl{\svar\mvar, \menv, \mstore, \makont}_\mtime,\mktab$
    &
    $\tpl{\mval, \mstore,\makont}_\mtimealt,\mktab$ if $\mval \in \mstore(\menv(\mvar))$
    \\
    $\tpl{\sapp{\mexpri0}{\mexpri1},\menv,\mstore,\makont}_\mtime,\mktab$
    &
    $\tpl{\mexpri0,\menv,\mstore,\kcons{\appl{\mexpri1,\menv}}{\mctx}}_\mtimealt,\mktab'$ \\
    where & $\mctx = \tpl{\sapp{\mexpri0}{\mexpri1},\menv,\mstore}_\mtime$ \\
          & $\mktab' = \joinm{\mktab}{\mctx}{\makont}$
    \\
    $\tpl{\mval,\mstore,\kcons{\appl{\mexpr,\menv'}}{\mctx}}_\mtime,\mktab$
    &
    $\tpl{\mexpr,\menv',\mstore,\kcons{\appr{\mval}}{\mctx}}_\mtimealt,\mktab$
    \\
    $\tpl{\mval,\menv,\mstore,\kcons{\appr{\slam{\mvar}{\mexpr},\menv'}}{\mctx}}_\mtime,\mktab$
    &
    $\tpl{\mexpr,\menv'',\mstore',\makont}_\mtimealt,\mktab$ if $\makont \in \mktab(\mctx)$ \\
    where & $\menv'' = \extm{\menv'}{\mvar}{\maddr}$ \\
          & $\mstore' = \joinm{\mstore}{\maddr}{\mval}$
  \end{tabular}
  \caption{$\CESKKstart$ semantics}
  \label{fig:ceskkstart-semantics}
\end{figure}

\iftr{
\subsection{Correctness}

The high level argument for correctness exploits properties of both machines.
Where the stack is unbounded (call this $\CESKt$), if every state in a trace shares a common tail in their continuations, that tail is \emph{irrelevant}.
This means the tail can be replaced with anything and still produce a valid trace.
We call this property more generally, ``context irrelevance.''
The $\CESKKstart$ machine maintains an invariant on $\mktab$ that says that $\makont \in \mktab(\mctx)$ represents a trace in $\CESKt$ that starts at the base of $\makont$ and reaches $\mctx$ with $\makont$ on top.
We can use this invariant and context irrelevance to translate steps in the $\CESKKstart$ machine into steps in $\CESKt$.
The other way around, we use a proposition that a full stack is represented by $\mktab$ via unrolling and follow a simple simulation argument.

The common tail proposition we will call $\hastail$ and the replacement function we will call $\replacetail$; they both have obvious inductive and recursive definitions respectively.
The invariant is stated with respect to the entire program, $\mexpr_\mathit{pgm}$:
\begin{mathpar}
  \inferrule{ }{\invmktab(\bot)} \quad
  \inferrule{\invmktab(\mktab) \\
      \forall \makont_c \in K. \startstate(\makont_c) \stepto_\CESKt^* \tpl{\mexpr_c,\menv_c,\mstore_c,\append{\mkont_c}{\epsilon}}_{\mtime_c}}
            {\invmktab(\extm{\mktab}{\tpl{\mexpr_c,\menv_c,\mstore_c}_{\mtime_c}}{K})} \\

  \inferrule{
    \startstate(\makont) \stepto_\CESKt^* \tpl{\mexpr,\menv,\mstore,\append{\makont}{\epsilon}}_\mtime \\
    \invmktab(\mktab)}
    {\inv(\tpl{\mexpr,\menv,\mstore,\makont}_\mtime,\mktab)}
  \end{mathpar}
where
\begin{align*}
 \startstate(\epsilon) &= \tpl{\mexpr_\mathit{pgm},\bot,\bot,\epsilon}_{\mtime_0} \\
                \startstate(\kcons{\mkframe}{\tpl{\mexpr_c,\menv_c,\mstore_c}_{\mtime_c}}) &=
                \tpl{\mexpr_c,\menv_c,\mstore_c,\epsilon}_{\mtime_c}
\end{align*}
We use $\append{\cdot}{\epsilon}$ to treat $\mctx$ like $\epsilon$ and construct a continuation in $\Kont$ rather than $\sa{Kont}$.
\begin{lemma}[Context irrelevance]\label{lem:irrelevance}
  For all traces $\mtrace \in \CESKt^*$ and continuations $\mkont$ such that $\hastail(\mtrace,\mkont)$, for any $\mkont'$, $\replacetail(\mtrace,\mkont,\mkont')$ is a valid trace.
\end{lemma}
\begin{proof}
  Simple induction on $\mtrace$ and cases on $\stepto_{\CESKt}$.
\end{proof}
\begin{lemma}[$\CESKKstart$ Invariant]\label{lem:invariant}
  For all $\mstate,\mstate' \in \sa{State}$, if $\inv(\mstate)$ and $\mstate \stepto \mstate'$, then $\inv(\mstate')$
\end{lemma}
\begin{proof}
  Routine case analysis.
\end{proof}
Note that the injection of $\mexpr_\mathit{pgm}$ into $\CESKKstart$, $(\tpl{\mexpr_\mathit{pgm},\bot,\bot,\epsilon}_{\mtime_0},\bot)$, trivially satisfies $\inv$.

The unrolling proposition is the following
\begin{mathpar}
  \inferrule{ }{\epsilon \in \unroll{\mktab}{\epsilon}} \quad
  \inferrule{\makont \in \mktab(\mctx),
             \mkont \in \unroll{\mktab}{\makont}}
            {\kcons{\mkframe}{\mkont} \in \unroll{\mktab}{\kcons{\mkframe}{\mctx}}}
\end{mathpar}
\begin{theorem}[Correctness]\label{thm:pushdown-correct}
  For all expressions $\mexpr_\mathit{pgm}$,
  \begin{itemize}
  \item{{\bf Soundness: } 
        if $\mstate \stepto_{\CESKt} \mstate'$,
        $\inv(\mstate\set{\mkont := \makont},\mktab)$,
        and $\mkont \in \unroll{\mktab}{\makont}$, then
        there are $\mktab',\makont'$ such that
        $\mstate\set{\mkont := \makont},\mktab \stepto_{\CESKKstart} \mstate'\set{\mkont := \makont'},\mktab'$ and $\mkont' \in \unroll{\mktab'}{\makont'}$}
  \item{{\bf Completeness:} if $\mastate,\mktab \stepto_{\CESKKstart} \mastate',\mktab'$
      and $\inv(\mastate,\mktab)$,
      for all $\mkont$, if $\mkont \in \unroll{\mktab}{\mastate.\makont}$ then
      there is a $\mkont'$ such that
      $\mastate\set{\makont := \mkont} \stepto_{\CESKt}
       \mastate'\set{\makont := \mkont'}$ and
       $\mkont' \in \unroll{\mktab}{\mastate'.\makont}$.}
  \end{itemize}
\end{theorem}
}

\paragraph{Revisiting the example}

First we consider what \zcfa{} gives us, to see where pushdown analysis improves.
The important difference is that in \kcfa{}, return points are stored in an address that is linked to the textual location of the function call, plus a $k$-bounded amount of calling history.
So, considering the common $k = 0$, the unknown function call within map (either \texttt{render{-}int} or \texttt{fact}) returns from the context of the second call to \texttt{map} to the context of the first call to \texttt{map}.
The non-tail calls aren't safe from imprecise return flow either.
In \texttt{map}, the self call returns directly to both external calls of \texttt{map}, bypassing the \texttt{cons} calling context.
An exact representation of the stack guards against such nonsensical predictions.

In our presentation, return points are stored in an address that represents the \emph{exact} calling context with respect to the abstract machine's components.
This means when there is a ``merging'' of return points, it really means that two places in the program have requested the exact same thing of a function, even with the same global values.
The function \emph{will} return to both places.
The predicted control flow in the example is as one would expect, or \emph{hope}, an analysis would predict: the correct flow.

\subsection{Engineered semantics for efficiency}\label{sec:eng-frontier}
We cover three optimizations that may be employed to accelerate the fixed-point computation.
\begin{enumerate}
\item{\label{item:chunk}Continuations can be ``chunked'' more coarsely at function boundaries instead of at each frame in order to minimize table lookups.}
\item{We can globalize $\mktab$ with no loss in precision, unlike a global store;%
      it will not need to be stored in the frontier but will need to be tracked by seen states.
      The seen states only need comparison, and a global $\mktab$ increases monotonically, so we can use Shivers' timestamp technique~\citep{ianjohnson:Shivers:1991:CFA}.
      The timestamp technique does not store an entire $\mktab$ in the seen set at each state, but rather how many times the global $\mktab$ has increased.}
\item{Since evaluation is the same regardless of the stack, we can memoize results to short-circuit to the answer.
      The irrelevance of the stack then precludes the need for timestamping the global $\mktab$.}
\end{enumerate}
This last optimization will be covered in more detail in \autoref{sec:memo}.
From here on, this paper will not explicitly mention timestamps.

A secondary motivation for the representation change in \ref{item:chunk} is that flow analyses commonly split control-flow graphs at function call boundaries to enable the combination of intra- and inter-procedural analyses.
In an abstract machine, this split looks like installing a continuation prompt at function calls.
We borrow a representation from literature on delimited continuations~\citep{ianjohnson:Biernacki2006274} to split the continuation into two components: the continuation and meta-continuation.
Our delimiters are special since each continuation ``chunk'' until the next prompt has bounded length.
The bound is roughly the deepest nesting depth of an expression in functions' bodies.
Instead of ``continuation'' and ``meta-continuation'' then, we will use terminology from CFA2 and call the top chunk a ``local continuation,'' and the rest the ``continuation.''\footnote{Since the continuation is either $\epsilon$ or a context, CFA2 calls these ``entries'' to mean execution entry into the program ($\epsilon$) or a function ($\mctx$). One can also understand these as entries in a table ($\mktab$). We stay with the ``continuation'' nomenclature because they represent full continuations.}

\autoref{fig:pushdown-vis} has a visualization of a hypothetical state space.
Reduction relations can be thought of as graphs: each state is a node, and if a state $\mstate$ reduces to $\mstate'$, then there is an edge $\mstate \stepto \mstate'$.
We can also view our various environments that contain pointers (addresses, contexts) as graphs: each pointer is a node, and if the pointer $\mctx$ references an object $\mlkont$ that contains another pointer $\mctx'$, then there is a labeled edge $\mctx \xrightarrow{\mlkont} \mctx'$.
States' contexts point into $\mktab$ to associate each state with a \emph{regular language} of continuations.
The reversed $\mktab$ graph can be read as a group of finite state machines that accepts all the continuations that are possible at each state that the reversed pointers lead to.
The $\kmt$ continuation is this graph's starting state.

\begin{figure}
  \centering
  \begin{tabular}{rlrl}
    $\mastate \in \sa{CESIK}$ &\hspace{-3mm}$= \tpl{\mexpr,\menv,\mstore,\mlkont,\makont}$& $\mlkont \in \LKont$ &\hspace{-3mm}$= \Frame^*$ \\
    & & $\makont \in \Kont$ &\hspace{-3mm}$::= \epsilon \alt \mctx$
  \end{tabular}
  \caption{$\CESIKKstar$ semantic spaces}
  \label{fig:pushdown-spaces}
\end{figure}

\begin{figure}
  \centering
  \includegraphics[scale=0.65]{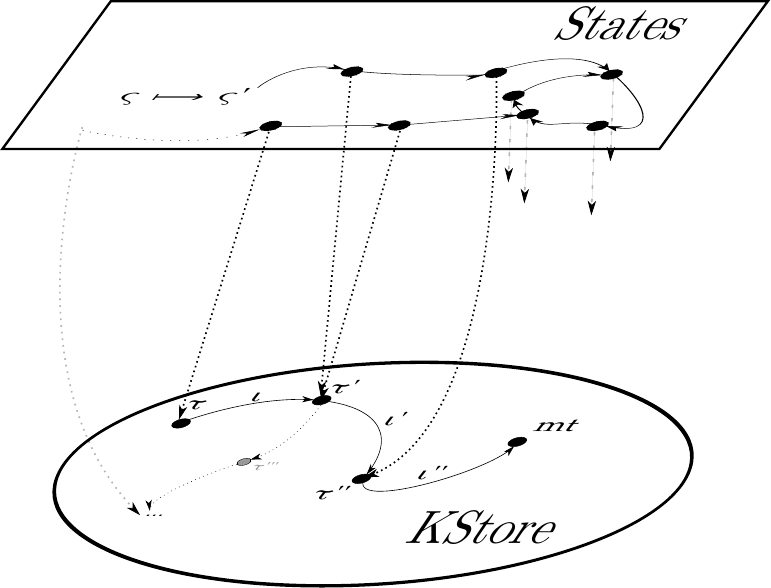}
  \caption{Graph visualization of states and $\mktab$}
  \label{fig:pushdown-vis}
\end{figure}

The resulting shuffling of the semantics to accommodate this new representation is in \autoref{fig:cesikkstar-semantics}.
The extension to $\mktab$ happens in a different rule -- function entry -- so the shape of the context changes to hold the function, argument, and store.
We have a choice of whether to introduce an administrative step to dereference $\mktab$ once $\mlkont$ is empty, or to use a helper metafunction to describe a ``pop'' of both $\mlkont$ and $\mkont$.
Suppose we choose the second because the resulting semantics has a 1-to-1 correspondence with the previous semantics.
A first attempt might land us here:
\begin{align*}
  \pop(\kcons{\mkframe}{\mlkont},\makont,\mktab) &= \set{(\mkframe,\mlkont,\makont)} \\
  \pop(\epsilon,\mctx,\mktab) &= \setbuild{(\mkframe,\mlkont,\makont)}{(\kcons{\mkframe}{\mlkont}, \makont) \in \mktab(\mctx)}
\end{align*}
However, tail calls make the dereferenced $\mctx$ lead to $(\epsilon,\mctx')$.
Because abstraction makes the store grow monotonically in a finite space, it's possible that $\mctx' = \mctx$ and a naive recursive definition of $\pop$ will diverge chasing these contexts.
Now $\pop$ must save all the contexts it dereferences in order to guard against divergence.
So $\pop(\mlkont,\makont,\mktab) = \popaux(\mlkont,\makont,\mktab,\emptyset)$ where
\begin{align*}
  \popaux(\epsilon,\epsilon,\mktab,G) &= \emptyset \\
  \popaux(\kcons{\mkframe}{\mlkont},\makont,\mktab,G) &= \set{(\mkframe,\mlkont,\makont)} \\
  \popaux(\epsilon,\mctx,\mktab,G) &= \setbuild{(\mkframe,\mlkont,\makont)}{(\kcons{\mkframe}{\mlkont}, \makont) \in \mktab(\mctx)} \\
  &\cup \bigcup\limits_{\mctx' \in G'}\popaux(\epsilon,\mctx',\mktab,G\cup G') \\
  \text{where } G' &= \setbuild{\mctx'}{(\epsilon,\mctx') \in \mktab(\mctx)} \setminus G
\end{align*}

In practice, one would not expect $G$ to grow very large.
Had we chosen the first strategy, the issue of divergence is delegated to the machinery from the fixed-point computation.\footnote{CFA2 employs the first strategy and calls it ``transitive summaries.''}
However, when adding the administrative state, the ``seen'' check requires searching a far larger set than we would expect $G$ to be.

\begin{figure}
  \centering
  $\mastate,\mktab \stepto \mastate',\mktab'$ \quad $\maddr = \alloc(\mastate,\mktab)$ \\
  \begin{tabular}{r|l}
    \hline\vspace{-3mm}\\
    $\tpl{\svar\mvar, \menv, \mstore, \mlkont, \makont},\mktab$
    &
    $\tpl{\mval,\mstore,\mlkont,\makont},\mktab$ if $\mval \in \mstore(\menv(\mvar))$
    \\
    $\tpl{\sapp{\mexpri0}{\mexpri1},\menv,\mstore,\mlkont,\makont},\mktab$
    &
    $\tpl{\mexpri0,\menv,\mstore,\kcons{\appl{\mexpri1,\menv}}{\mlkont},\makont},\mktab$
    \\
    $\tpl{\mval, \mstore, \mlkont,\makont},\mktab$
    &
    $\tpl{\mexpr,\menv',\mstore,\kcons{\appr{\mval,\menv}}{\mlkont'},\makont'},\mktab$ \\
    &
    if $\appl{\mexpr,\menv'},\mlkont',\makont' \in \pop(\mlkont,\makont,\mktab)$ \\
    $\tpl{\mval,\mstore, \mlkont,\makont},\mktab$
    &
    $\tpl{\mexpr,\extm{\menv}{\mvar}{\maddr},\mstore',\epsilon,\mctx},\mktab'$ \\
    & if $\appr{\slam{\mvar}{\mexpr},\menv}, \mlkont', \makont' \in \pop(\mlkont,\makont,\mktab)$ \\
    where & $\mstore' = \joinm{\mstore}{\maddr}{\mval}$ \\
    & $\mctx = (\tpl{\slam{\mvar}{\mexpr},\menv},\mval,\mstore)$ \\
    & $\mktab' = \joinm{\mktab}{\mctx}{(\mlkont,\makont)}$
  \end{tabular} \\
  \caption{$\CESIKKstar$ semantics}
  \label{fig:cesikkstar-semantics}
\end{figure}

\begin{align*}
  {\mathcal F}_{\mexpr}(S,R,F,\mktab) &= (S \cup F, R \cup R', F'\setminus S, \mktab') \\
  I &= \bigcup\limits_{s=(\mastate,\mstore) \in F}{\setbuild{(\tpl{\mastate,\mastate'}, \mktab')}{\mastate,\mktab \stepto \mastate',\mktab'}} \\
  R' &= \pi_0 I \qquad F' = \pi_1 R' \qquad \mktab' = \bigsqcup\pi_1 I
\end{align*}
For a program $\mexpr$, we will say $(\emptyset,\emptyset,\set{\tpl{\mexpr,\bot,\bot,\epsilon,\epsilon}},\bot)$ is the bottom element of ${\mathcal F}_{\mexpr}$'s domain.
The ``analysis'' then is then the pair of the $R$ and $\mktab$ components of $\lfp({\mathcal F}_{\mexpr})$.

\iftr{
\paragraph{Correctness} The correctness argument for this semantics is not about single steps but instead about the entire relation that ${\mathcal F}$ computes.
The argument is that the $R$ and $\mktab$ components of the system represent a slice of the unbounded relation $\stepto_{\CESK}$ (restricted to reachable states).
We will show that traces in any $n \in \nat$ times we \emph{unfold} $\stepto_{\CESK}$ from the initial state, there is a corresponding $m$ applications of ${\mathcal F}$ that reify into a relation that exhibit the same trace.
Conversely, any trace in the reification of ${\mathcal F}_{\mexpr}^m(\bot)$ has the same trace in some $n$ unfoldings of $\stepto_{\CESK}$.
For an arbitrary $\alloc$ function, we cannot expect ${\mathcal F}$ to have a fixed point, so this property is the best we can get.
For a finite $\alloc$ function, Kleene's fixed point theorem dictates there is a $m$ such that ${\mathcal F}_{\mexpr}^m(\bot)$ is a fixed point, so every trace in the reified relation is also a trace in an unbounded number of unfoldings of $\stepto_{\CESK}$.
This is the corresponding completeness argument for the algorithm.

\begin{align*}
  \reachrestrict(\mstate_0, \stepto, 0) &= \setbuild{(\mstate_0,\mstate)}{\mstate_0 \stepto \mstate} \\
  \reachrestrict(\mstate_0, \stepto, n+1) &= \stepextend(\reachrestrict(\mstate_0,\stepto,n)) \\
  \textit{where } \stepextend(R) &= R \cup \setbuild{(\mstate,\mstate')}{(\_,\mstate) \in R, \mstate \stepto \mstate'}
\end{align*}
The reification simply realizes all possible complete continuations that a state could have, given $\mktab$:
\begin{mathpar}
  \inferrule{
  \tpl{\tpl{\mexpr,\menv,\mstore,\makont},
      \tpl{\mexpr',\menv',\mstore',\makont'}} \in R \\
  \mkont \in \tails{\mktab}{\makont}}
  {\tpl{\mexpr,\menv,\mstore,\append{\makont}{\mkont}} \stepto_{\reify(S,R,F,\mktab)}
   \tpl{\mexpr',\menv',\mstore'\append{\makont'}{\mkont}}}
\end{mathpar}
The additional judgment about tails is a variant of $\mathit{unroll}$ in order to get a common tail:
\begin{mathpar}
  \inferrule{ }{\epsilon \in \tails{\mktab}{\epsilon}} \quad
  \inferrule{\makont \in \mktab(\mctx) \\
             \mkont \in \unroll{\mktab}{\makont}}
            {\mkont \in \tails{\mktab}{\kcons{\mkframe}{\mctx}}}
\end{mathpar}
\begin{theorem}[Correctness]\label{thm:global-pushdown}
  For all $\mexpr_0$, let $\mstate_0 = \tpl{\mexpr_0,\bot,\bot,\epsilon}$ in
  $\forall n \in \nat, \mstate,\mstate' \in \CESK$:
  \begin{itemize}
  \item{if $(\mstate,\mstate') \in \reachrestrict(\mstate_0,\stepto_{\CESK},n)$ then
      there is an $m$ such that $\mstate \stepto_{\reify({\mathcal F}_{\mexpr_0}^m(\bot))} \mstate'$}
  \item{if $\mstate \stepto_{\reify({\mathcal F}_{\mexpr_0}^n(\bot))} \mstate'$ then
      there is an $m$ such that $(\mstate,\mstate') \in \reachrestrict(\mstate_0,\stepto_{\CESK},m)$}
  \end{itemize}
\end{theorem}
\begin{proof}
  By induction on $n$.
\end{proof}
}

\subsection{Remarks about cost}

The common tradeoff for performance over precision is to use a global store.
A representation win originally exploited by Shivers~\citep{ianjohnson:Shivers:1991:CFA} is to represent the seen states' stores by the \emph{age} of the store.
A context in this case contains the store age for faster comparison.
Old stores are mostly useless since the current one subsumes them, so a useful representation for the seen set is as a map from the \emph{rest of the state} to the store age it was last visited with.
We will align with the analysis literature and call these ``rest of state'' objects \emph{points}.
Note that since the store age becomes part of the state representation due to ``context,'' there are considerably more points than in the comparable finite state approach.
When we revisit a state because the store age (or $\mktab$ age) is different from the last time we visited it (hence we're visiting a \emph{new state}), we can clobber the old ages.
A finite state approach will use less memory because the seen set will have a smaller domain (fewer distinctions made because of the lack of a ``context'' component).
%
%

%% file: inspection.tex
Since we just showed how to produce a pushdown system from an abstract machine, some readers may be concerned that we have lost the ability to reason about the stack as a whole.
This is not the case.
The semantics may still refer to $\mktab$ to make judgments about the possible stacks that can be realized at each state.
The key is to interpret the functions making these judgments again with the AAM methodology.

Some semantic features allow a language to inspect some arbitrarily deep part of the stack, or compute a property of the whole stack before continuing.
Java's access control security features are an example of the first form of inspection, and garbage collection is an example of the second.
We will demonstrate both forms are simple first-order metafunctions that the AAM methodology will soundly interpret.
Access control can be modeled with continuation marks, so we demonstrate with the CM machine of \citeauthor{dvanhorn:Clements2004Tailrecursive}.

Semantics that inspect the stack do so with metafunction calls that recur down the stack.
Recursive metafunctions have a semantics as well, hence fair game for AAM.
And, they should always terminate (otherwise the semantics is hosed).
We can think of a simple pattern-matching recursive function as a set of rewrite rules that apply repeatedly until it reaches a result.
Interpreted via AAM, non-deterministic metafunction evaluation leads to a set of possible results.

The finite restriction on the state space carries over to metafunction inputs, so we can always detect infinite loops that abstraction may have introduced and bail out of that execution path.
Specifically, a metafunction call can be seen as an initial state, $s$, that will evaluate through the metafunction's rewrite rules $\stepto$ to compute all terminal states (outputs):
\begin{align*}
  \terminal &: \forall A. \text{relation } A \times A \to \wp(A) \\
  \terminal(\stepto,s) &= \terminalaux(\emptyset,\set{s},\emptyset) \\[2pt]
  \text{where } \terminalaux(S, \emptyset, T) &= T \\
   \terminalaux(S, F, T) &= \terminalaux(S\cup F, F', T\cup T') \\
   & \text{where } T' = \bigcup\limits_{s \in F}{\post(s) \deceq\emptyset \to \set{s},\emptyset} \\
             &\phantom{\text{where }} F' = \bigcup\limits_{s\in F}{\post(s)} \setminus S \\
                   \post(s) &= \setbuild{s'}{s \stepto s'}
\end{align*}

This definition is a typical worklist algorithm.
It builds the set of terminal terms, $T$, by exploring the frontier (or worklist), $F$, and only adding terms to the frontier that have not been seen, as represented by $S$.
If $s$ has no more steps, $\post(s)$ will be empty, meaning $s$ should be added to the terminal set $T$.%
\iftr{
We prove a correctness condition that allows us to reason equationally with $\terminal$ later on:
\begin{lemma}[$\terminalaux$ correct]\label{lem:term-correct}
  Fix $\stepto$. Constrain arbitrary $S, F,T$ such that $T \sqsubseteq S$ and $\forall s\in T, \post(s) = \emptyset$, $F \cap S = \emptyset$, and for all $s \in S$, $\post(s) \subseteq S \cup F$.
  \begin{itemize}
  \item{\textbf{Soundness:} for all $s \in S \cup F$, if $s \stepto^* s_t$ and $\post(s_t) = \emptyset$ then $s_t \in \terminalaux(S,F,T)$.}
  \item{\textbf{Completeness:} for all $s \in \terminalaux(S,F,T)$ there is an $s_0 \in S \cup F$ such that $s_0 \stepto^* s$ and $\post(s) = \emptyset$.}
  \end{itemize}
\end{lemma}
\begin{proof}
  By induction on $\terminalaux$'s recursion scheme.
\end{proof}
}
Note that it is possible for metafunctions' rewrite rules to themselves use metafunctions, but the \emph{seen} set ($S$) for $\terminal$ must be dynamically bound\footnote{This is a reference to dynamic scope as opposed to lexical scope.} -- it cannot restart at $\emptyset$ upon reentry.
Without this precaution, the host language will exceed its stack limits when an infinite path is explored, rather than bail out.

\subsection{Case study for stack traversal: GC}\label{sec:gc}
Garbage collection is an example of a language feature that needs to crawl the stack, specifically to find live addresses.
We are interested in garbage collection because it can give massive precision boosts to analyses~\citep{dvanhorn:Might:2006:GammaCFA,dvanhorn:Earl2012Introspective}.
Unadulterated, abstract GC inflicts an exponential state space that can destroy performance.
The following function will produce the set of live addresses in the stack:

\begin{align*}
  \kontlive &: \Frame^* \to \wp(\Addr) \\
  \kontlive(\mkont) &= \kontliveaux(\mkont,\emptyset) \\[2pt]
  \kontliveaux(\epsilon,L) &= L \\
  \kontliveaux(\kcons{\mkframe}{\mkont}, L) &= \kontliveaux(\mkont, L\cup\touches(\mkframe)) \\
  \text{where } \touches(\appl{\mexpr,\menv}) &= \touches(\appr{\mexpr,\menv}) = \touches(\mexpr,\menv) \\
                \touches(\mexpr,\menv) &= \setbuild{\menv(\mvar)}{\mvar \in \fv(\mexpr)}
\end{align*}

When interpreted via AAM, the continuation is indirected through $\mktab$ and leads to multiple results, and possibly loops through $\mktab$.
Thus this is more properly understood as
\begin{align*}
  \kontlive(\mktab,\makont) &= \terminal(\stepto, \kontliveaux(\mktab,\makont,\emptyset)) \\[2pt]
  \kontliveaux(\mktab,\epsilon,L) &\stepto L \\
  \kontliveaux(\mktab,\kcons{\mkframe}{\mctx}, L) &\stepto \kontliveaux(\mktab,\makont, L\cup\touches(\mkframe)) \text{ if } \makont \in \mktab(\mctx)
\end{align*}

A garbage collecting semantics can choose to collect the heap with respect to each live set (call this $\Gamma^*$), or, soundly, collect with respect to their union (call this $\hat\Gamma$).\footnote{The garbage collecting version of PDCFA~\citep{ianjohnson:DBLP:journals/jfp/JohnsonSEMH14} evaluates the $\hat\Gamma$ strategy.}
On the one hand we could have tighter collections but more possible states, and on the other hand we can leave some precision behind in the hope that the state space will be smaller.
In the general idea of relevance versus irrelevance, the continuation's live addresses are relevant to execution, but are already implicitly represented in contexts because they must be mapped in the store's domain.

A state is ``collected'' only if live addresses remain in the domain of $\mstore$.
We say a value $\mval \in \mstore(\maddr)$ is live if $\maddr$ is live.
If a value is live, any addresses it touches are live; this is captured by the computation in $\reaches$:
\begin{align*}
  \reaches(\mathit{root},\mstore) &=
 \setbuild{\maddralt}{\maddr \in \mathit{root}, \maddr \leadsto_\mstore^* \maddralt} \\
&  \infer{\mval \in \mstore(\maddr) \\ \maddralt \in \touches(\mval)}{\maddr \leadsto_\mstore \maddralt}
\end{align*}
So the two collection methods are as follows.
Exact GC produces different collected states based on the possible stacks' live addresses:\footnote{It is possible and more efficient to build the stack's live addresses piecemeal as an additional component of each state, precluding the need for $\kontlive$. Each stack in $\mktab$ would also store the live addresses to restore on pop.}
\begin{align*}
  \Gamma^*(\mastate,\mktab) &=
    \setbuild{\mastate\set{\mstore:=\mastate.\mstore|_L}}{L \in \live^*(\mastate,\mktab)} \\
  \live^*(\tpl{\mexpr,\menv,\mstore,\makont},\mktab) &=
    \setbuild{\reaches(\touches(\mexpr,\menv) \cup L, \mstore)}{L \in \kontlive(\mktab,\makont)}
\end{align*}
\begin{equation*}
  \infer{\mastate,\mktab \stepto \mastate',\mktab' \\
         \mastate' \in \Gamma^*(\mstate',\mktab')}
        {\mastate,\mktab \stepto_{\Gamma^*} \mastate',\mktab}
\end{equation*}
And inexact GC produces a single state that collects based on all (known) stacks' live addresses:
\begin{align*}
  \hat\Gamma(\mastate,\mktab) &=
  \mastate\set{\mstore:=\mastate.\mstore|_{\widehat{\live}(\mastate,\mktab)}} \\
  \widehat{\live}(\tpl{\mexpr,\menv,\mstore,\makont},\mktab) &=
    \reaches(\touches(\mexpr,\menv) \cup \bigcup\kontlive(\mktab,\makont), \mstore)
\end{align*}
\begin{equation*}
  \infer{\mastate,\mktab \stepto \mastate',\mktab'}
        {\mastate,\mktab \stepto_{\hat\Gamma} \hat\Gamma(\mstate',\mktab'),\mktab'}
\end{equation*}

Without the continuation store, the baseline GC is
\begin{align*}
  \Gamma(\mstate) &= \mstate\set{\mstore:=\mstate.\mstore|_{\live(\mstate)}} \\
  \live(\mexpr,\menv,\mstore,\mkont) &= \reaches(\touches(\mexpr,\menv)\cup \kontlive(\mkont), \mstore)
\end{align*}
\begin{equation*}
  \infer{\mstate \stepto \mstate'}
        {\mstate \stepto_{\Gamma} \Gamma(\mstate')}  
\end{equation*}
Suppose at arbitrary times we decide to perform garbage collection rather than continue with garbage.
So when $\mastate \stepto \mastate'$, we instead do $\mastate \stepto_\Gamma \mastate'$.
The times we perform GC do not matter for soundness, since we are not analyzing GC behavior.
However, garbage stands in the way of completeness.
Mismatches in the GC application for the different semantics lead to mismatches in resulting state spaces, not just up to garbage in stores, but in spurious paths from dereferencing a reallocated address that was not first collected.

\iftr{
The state space compaction that continuation stores give us makes ensuring GC times match up for the completeness proposition tedious.
Our statement of completeness then will assume both semantics perform garbage collection on every step.

\begin{lemma}[Correctness of $\kontlive$]
  For all $\mktab,\mkont,\makont,L$,
  \begin{itemize}
  \item{\textbf{Soundness:} if $\mkont \in \unroll{\mktab}{\makont}$ then $\kontliveaux(\mkont,L) \in \terminal(\stepto,\kontliveaux(\mktab,\makont,L))$}
  \item{\textbf{Completeness:} for all $L' \in \kontliveaux(\mktab,\makont,L)$ there is a $\mkont \in \unroll{\mktab}{\makont}$ such that $L' = \kontliveaux(\mkont,L)$.}
  \end{itemize}
\end{lemma}
\begin{proof}
  Soundness follows by induction on the unrolling. Completeness follows by induction on the trace from completeness in \autoref{lem:term-correct}.
\end{proof}
\begin{theorem}[Correctness of $\Gamma^*\CESKKstart$]
  For all expressions $\mexpr_0$,
  \begin{itemize}
  \item{{\bf Soundness: } 
        if $\mstate \stepto_{\Gamma\CESKt} \mstate'$,
        $\inv(\mstate\set{\mkont := \makont},\mktab)$,
        and $\mstate.\mkont \in \unroll{\mktab}{\makont}$, then
        there are $\mktab',\makont',\mstore'$ such that
        $\mstate\set{\mkont := \makont},\mktab \stepto_{\Gamma^*\CESKKstart} \mastate',\mktab'$ where
        $\mastate' = \mstate'\set{\mkont := \makont',\mstore:=\mstore'}$ 
        and $\mstate'.\mkont \in \unroll{\mktab'}{\makont'}$
        and finally there is an $L \in \live^*(\mastate',\mktab')$ such that $\mstore'|_L = \mstate'.\mstore|_{\live(\mstate')}$}
  \item{{\bf Completeness:} if $\mastate\equiv\tpl{\mexpr,\menv,\mstore,\makont},\mktab \stepto_{\Gamma^*\CESKKstart} \mastate',\mktab'$ and there is an $L_\mkont \in \kontlive(\mktab,\makont)$ such that $\mstore|_L = \mstore$ (where $L = \reaches(\touches(\mexpr,\menv) \cup L_\mkont, \mstore)$) and $\inv(\mastate,\mktab)$,
      for all $\mkont \in \unroll{\mktab}{\makont}$ such that $\kontlive(\mkont) = L_\mkont$,
      there is a $\mkont'$ such that
      $\mastate\set{\makont := \mkont} \stepto_{\Gamma\CESKt}
      \mastate'\set{\makont := \mkont'}$ (a GC step) and
      $\mkont' \in \unroll{\mktab}{\mastate'.\makont}$}
  \end{itemize}  
\end{theorem}

\begin{theorem}[Soundness of $\hat\Gamma\CESKKstart$]
  For all expressions $\mexpr_0$,
  if $\mstate \stepto_{\Gamma\CESKt} \mstate'$,
  $\inv(\mstate\set{\mkont := \makont},\mktab)$,
  and $\mstate.\mkont \in \unroll{\mktab}{\makont}$, then
  there are $\mktab',\makont',\mstore''$ such that
  $\mstate\set{\mkont := \makont},\mktab \stepto_{\hat\Gamma\CESKKstart} \mastate',\mktab'$ where
  $\mastate' = \mstate'\set{\mkont := \makont',\mstore:=\mstore''}$
  and $\mstate'.\mkont \in \unroll{\mktab'}{\makont'}$
  and finally $\mstate'.\mstore|_{\live(\mstate')} \sqsubseteq \mstore''|_{\widehat{live}(\mastate',\mktab')}$
\end{theorem}
The proofs are straightforward, and use the usual lemmas for GC, such as idempotence of $\Gamma$ and $\touches \subseteq \reaches$.
}
\subsection{Case study analyzing security features: the CM machine}
The CM machine provides a model of access control: a dynamic branch of execution is given permission to use some resource if a continuation mark for that permission is set to ``grant.''
There are three new forms we add to the lambda calculus to model this feature: {\tt grant}, {\tt frame}, and {\tt test}.
The {\tt grant} construct adds a permission to the stack.
The concern of unforgeable permissions is orthogonal, so we simplify with a set of permissions that is textually present in the program:
\begin{align*}
  \mperm \in \Permissions & \text{ a set} \\
  \Expr &::= \ldots \alt \sgrant{\mperm}{\mexpr} \alt \sframe{\mperm}{\mexpr} \alt \stest{\mperm}{\mexpr}{\mexpr}
\end{align*}
Permissions have the \emph{granted} or \emph{denied} status saved in a permission map within the continuation.
Each stack frame added to the continuation carries the permission map.
The empty continuation also carries a permission map.
\begin{align*}
  \mpermmap \in \PermissionMap &= \Permissions \finto \GD \\
  \mgd \in \GD &::= \Grant \alt \Deny \\
  \mkont \in \Kont &::= \epsilon^\mpermmap \alt \kconsm{\mkframe}{\mpermmap}{\mkont}
\end{align*}
The {\tt frame} construct ensures that only a given set of permissions are allowed, but not necessarily granted.
The security is in the semantics of {\tt test}: we can test if all marks in some set $\mperm$ have been granted in the stack without first being denied; this involves crawling the stack:
\begin{align*}
  \OK(\emptyset,\mkont) &= \mathit{True} \\
  \OK(\mperm,\epsilon^\mpermmap) &= \passp(\mperm,\mpermmap) \\
  \OK(\mperm,\kconsm{\mkframe}{\mpermmap}{\mkont}) &= \passp(\mperm,\mpermmap) \textbf{ and } \OK(\mperm \setminus \mpermmap^{-1}(\Grant), \mkont) \\
  \text{where }\passp(\mperm,\mpermmap) &= \mperm \cap \mpermmap^{-1}(\Deny) \deceq \emptyset
\end{align*}
The set subtraction is to say that granted permissions do not need to be checked farther down the stack.
Continuation marks respect tail calls and have an interface that abstracts over the stack implementation.
Crucially, the added constructs do not add frames to the stack; instead, they update the permission map in the top frame, or if empty, the empty continuation's permission map.
Update for continuation marks:
\begin{align*}
  \extm{\mpermmap}{\mperm}{\mgd} &= \lambda x. x \decin \mperm \to \mgd, \mpermmap(x) \\
  \extm{\mpermmap}{\overline{\mperm}}{\mgd} &= \lambda x. x \decin \mperm \to \mpermmap(x),\mgd \end{align*}

\begin{figure}
  \centering
  \begin{tabular}{r|l}
    \hline\vspace{-3mm}\\
    $\tpl{\sgrant{\mperm}{\mexpr}, \menv, \mstore, \mkont}$
    &
    $\tpl{\mexpr,\menv,\mstore, \extm{\mkont}{\mperm}{\Grant}}$
    \\
    $\tpl{\sframe{\mperm}{\mexpr}, \menv,\mstore,  \mkont}$
    &
    $\tpl{\mexpr,\menv,\mstore, \extm{\mkont}{\overline{\mperm}}{\Deny}}$
    \\
    $\tpl{\stest{\mperm}{\mexpri0}{\mexpri1}, \menv,\mstore,  \mkont}$
    &
    $\tpl{\mexpri0,\menv,\mstore, \mkont}$ if $\mathit{True} = \OK(\mperm,\mkont)$
    \\
    &
    $\tpl{\mexpri1,\menv,\mstore, \mkont}$ if $\mathit{False} = \OK(\mperm,\mkont)$
  \end{tabular}
  \caption{CM machine semantics}
  \label{fig:cm-semantics}
\end{figure}

The abstract version of the semantics has one change on top of the usual continuation store.
The {\tt test} rules are now
\begin{align*}
  \tpl{\stest{\mperm}{\mexpri0}{\mexpri1}, \menv, \mstore, \makont},\mktab
  &\stepto
  \tpl{\mexpri0,\menv,\mstore, \makont},\mktab \text{ if } \mathit{True} \in \widehat{\OK}(\mktab,\mperm,\makont)
  \\
  &\stepto
  \tpl{\mexpri1,\menv,\mstore,\makont},\mktab \text{ if } \mathit{False} \in \widehat{\OK}(\mktab,\mperm,\makont)
\end{align*}
where the a new $\widehat{\OK}$ function uses $\terminal$ and rewrite rules:
\begin{align*}
  \widehat{\OK}(\mktab,\mperm,\makont) &= \terminal(\stepto,\widehat{\OK}^*(\mktab,\mperm,\makont)) \\[2pt]
  \widehat{\OK}^*(\mktab,\emptyset,\makont) &\stepto \mathit{True} \\
  \widehat{\OK}^*(\mktab,\mperm,\epsilon^\mpermmap) &\stepto \passp(\mperm,\mpermmap) \\
  \widehat{\OK}^*(\mktab,\mperm,\kconsm{\mkframe}{\mpermmap}{\mctx}) &\stepto b \text{ where }\\ &\phantom{\stepto} b \in \setbuild{\passp(\mperm,\mpermmap) \textbf{ and } b}{
          \makont \in \mktab(\mctx),
          \\&\phantom{\stepto b \in \{}b \in \widehat{\OK}(\mktab,\mperm\setminus\mpermmap^{-1}(\Grant),\makont))}
\end{align*}
This definition works fine with the reentrant $\terminal$ function with a dynamically bound \emph{seen} set, but otherwise needs to be rewritten to accumulate the Boolean result as a parameter of $\widehat{OK}^*$.

\iftr{
We use the accumulator version in the proof for simplicity.

\begin{lemma}[Correctness of $\widehat{\OK}$]
  For all $\mktab,\mperm,\mkont,\makont$,
  \begin{itemize}
  \item{\textbf{Soundness:} if $\mkont \in \unroll{\mktab}{\makont}$ then $\OK(\mperm,\mkont) \in \widehat{\OK}(\mktab,\mperm,\makont)$.}
  \item{\textbf{Completeness:} if $b \in \widehat{\OK}(\mktab,\mperm,\makont)$ then there is a $\mkont \in \unroll{\mktab}{\makont}$ such that $b = \OK(\mperm,\mkont)$.}
  \end{itemize}
\end{lemma}
With this lemma in hand, the correctness proof is almost identical to the core proof of correctness.
\begin{theorem}[Correctness]
  The abstract semantics is sound and complete in the same sense as \autoref{thm:pushdown-correct}.
\end{theorem}
}

%% file: shiftreset.tex
In \autoref{sec:pushdown} we showed how to get a pushdown abstraction by separating continuations from the heap that stores values.
This separation breaks down when continuations themselves become values via first-class control operators.
The glaring issue is that continuations become ``storeable'' and relevant to the execution of functions.
But, it was precisely the \emph{irrelevance} that allowed the separation of $\mstore$ and $\mktab$.
Specifically, the store components of continuations become elements of the store's codomain --- a recursion that can lead to an unbounded state space and therefore a non-terminating analysis.
We apply the AAM methodology to cut out the recursion; whenever a continuation is captured to go into the store, we allocate an address to approximate the store component of the continuation.
We introduce a new environment, $\mmktab$, that maps these addresses to the stores they represent.
The stores that contain addresses in $\mmktab$ are then \emph{open}, and must be paired with $\mmktab$ to be \emph{closed}.
This poses the same problem as before with contexts in storeable continuations.
Therefore, we give up some precision to regain termination by \emph{flattening} these environments when we capture continuations.
Fresh allocation still maintains the concrete semantics, but we necessarily lose some ability to distinguish contexts in the abstract.

\subsection{Case study of first-class control: shift and reset}
We choose to study {\tt shift} and {\tt reset}~\citep{ianjohnson:danvy:filinski:delim:1990} because delimited continuations have proven useful for implementing web servers~\citep{dvanhorn:Queinnec2004Continuations,jay-communication}, providing processes isolation in operating systems~\citep{dvanhorn:Kiselyov2007Delimited}, representing computational effects~\citep{dvanhorn:Filinski1994Representing}, modularly implementing error-correcting parsers~\citep{dvanhorn:Shivers2011Modular}, and finally undelimited continuations are \emph{pass\'e} for good reason~\citep{ianjohnson:kiselyov:against-callcc}.
Even with all their uses, however, their semantics can yield control-flow possibilities that surprise their users.
A \emph{precise} static analysis that illuminates their behavior is then a valuable tool.

Our concrete test subject is the abstract machine for shift and reset adapted from \citet{ianjohnson:Biernacki2006274} in the ``{\bf ev}al, {\bf co}ntinue'' style in \autoref{fig:shift-reset}.
The figure elides the rules for standard function calls.
The new additions to the state space are a new kind of value, $\vcomp{\mkont}$, and a \emph{meta-continuation}, $\mmkont \in \MKont = \Kont^*$ for separating continuations by their different prompts.
Composable continuations are indistinguishable from functions, so even though the meta-continuation is concretely a list of continuations, its conses are notated as function composition: $\mkapp{\mkont}{\mmkont}$.

\begin{figure}
  \centering
  $\mstate \stepto_\SR \mstate'$ \\
  \begin{tabular}{r|l}
    \hline
    $\ev{\sreset{\mexpr}, \menv, \mstore,\mkont, \mmkont}$
    &
    $\ev{\mexpr, \menv, \mstore,\epsilon, \mkapp{\mkont}{\mmkont}}$
    \\
    $\co{\epsilon, \mkapp{\mkont}{\mmkont}, \mval,\mstore}$
    &
    $\co{\mkont, \mmkont, \mval,\mstore}$
    \\
    $\ev{\sshift{\mvar}{\mexpr}, \menv, \mstore,\mkont, \mmkont}$
    &
    $\ev{\mexpr, \extm{\menv}{\mvar}{\maddr},\mstore',\epsilon,\mmkont}$
    \\
    where & $\mstore' = \joinm{\mstore}{\maddr}{\vcomp{\mkont}}$
    \\
    $\co{\kcons{\appr{\vcomp{\mkont'}}}{\mkont}, \mmkont, \mval,\mstore}$
    &
    $\co{\mkont', \mkapp{\mkont}{\mmkont}, \mval,\mstore}$
    \\
    \text{standard rules} \\
    $\ev{\svar\mvar, \menv, \mstore, \mkont,\mmkont}$
    &
    $\co{\mkont,\mmkont,\mval, \mstore}$ if $\mval \in \mstore(\menv(\mvar))$
    \\
    $\ev{\sapp{\mexpri0}{\mexpri1},\menv,\mstore,\mkont,\mmkont}$
    &
    $\ev{\mexpri0,\menv,\mstore,\kcons{\appl{\mexpri1,\menv}}{\mkont},\mmkont}$
    \\
    $\co{\kcons{\appl{\mexpr,\menv}}{\mkont},\mmkont,\mval,\mstore}$
    &
    $\ev{\mexpr,\menv,\mstore,\kcons{\appr{\mval}}{\mkont},\mmkont}$
    \\
    $\co{\kcons{\appr{\slam{\mvar}{\mexpr},\menv}}{\mkont},\mmkont,\mval,\mstore}$
    &
    $\ev{\mexpr,\menv',\mstore',\mkont,\mmkont}$ \\
    where & $\menv' = \extm{\menv}{\mvar}{\maddr}$ \\
          & $\mstore' = \joinm{\mstore}{\maddr}{\mval}$
  \end{tabular}  
  \caption{Machine semantics for shift/reset}
  \label{fig:shift-reset}
\end{figure}

\subsection{Reformulated with continuation stores}
The machine in \autoref{fig:shift-reset} is transformed now to have three new tables: one for continuations ($\mktab_{\makont}$), one as discussed in the section beginning to close stored continuations ($\mmktab$), and one for meta-continuations ($\mktab_{\mamkont}$).
The first is like previous sections, albeit continuations may now have the approximate form that is storeable.
The meta-continuation table is more like previous sections because meta-contexts are not storeable.
Meta-continuations do not have simple syntactic strategies for bounding their size, so we choose to bound them to size 0.
They could be paired with lists of $\sa{Kont}$ bounded at an arbitrary $n \in \nat$, but we simplify for presentation.

Contexts for continuations are still at function application, but now contain $\mmktab$.
Contexts for meta-continuations are in two places: manual prompt introduction via {\tt reset}, or via continuation invocation.
At continuation capture time, continuation contexts are approximated to remove $\mastore$ and $\mmktab$ components.
The different context spaces are thus:
\begin{align*}
  \msctx \in \ExactContext &::= \tpl{\mexpr,\menv,\mastore,\mmktab} \\
  \mactx \in \sa{Context} &::= \tpl{\mexpr,\menv,\maddr} \\
  \mctx \in \Context &::= \mactx \alt \msctx \\
  \mmctx \in \MContext &::= \tpl{\mexpr,\menv,\mastore,\mmktab}
                       \alt \tpl{\mvkont, \maval, \mastore, \mmktab} \\
\end{align*}

Revisiting the graphical intuitions of the state space, we have now $\mvkont$ in states' stores, which represent an \emph{overapproximation} of a set of continuations.
We augment the illustration from \autoref{fig:pushdown-vis} in \autoref{fig:shiftreset-vis} to include the new $\MKStore$ and the overapproximating behavior of $\mvkont$.
The informal notation $\mstore \leadsto \mvkont$ suggests that the state's store \emph{contains}, or \emph{refers to} some $\mvkont$.

\begin{figure}
  \centering
  \includegraphics[scale=0.6]{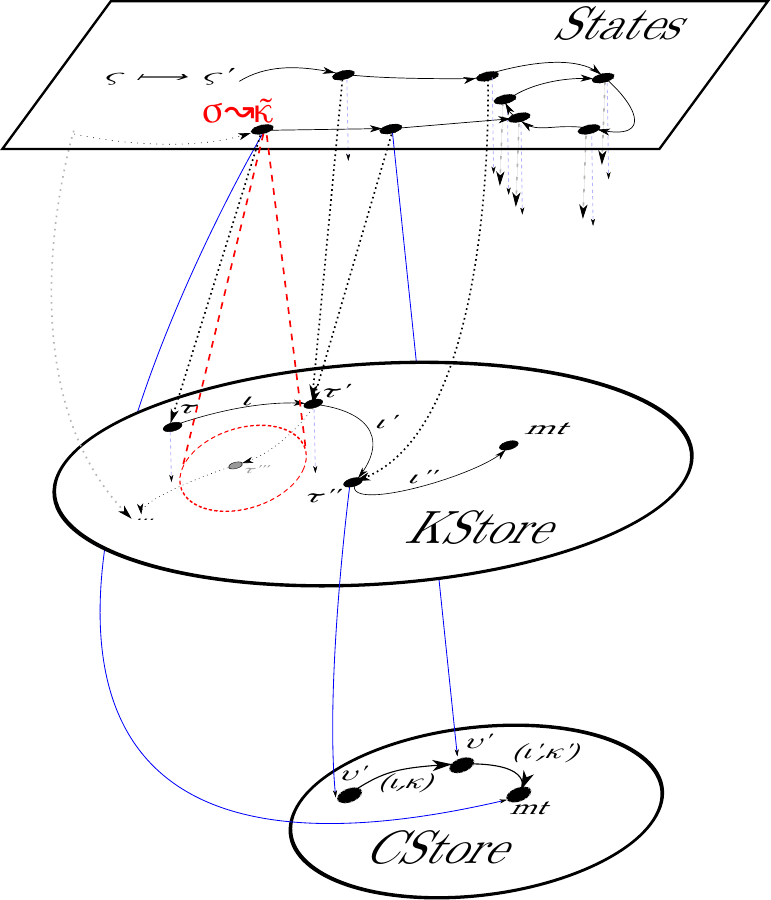}
  \caption{Graphical visualization of states, $\mktab_{\makont}$ and $\mktab_{\mamkont}$.}
  \label{fig:shiftreset-vis}
\end{figure}

\begin{figure}
  \centering
  \begin{tabular}{rlrl}
    $\mastate \in \sa{SR}$ &\multicolumn{3}{l}{\hspace{-3mm}$::= \ev{\mexpr,\menv,\mastore,\mmktab,\makont,\mamkont} \alt \co{\makont,\mamkont,\maval,\mastore,\mmktab}$} \\
    $\State$ & \multicolumn{3}{l}{\hspace{-3mm}$::= \mastate,\mktab_{\makont},\mktab_{\mamkont}$} \\
    $\mmktab \in \MKTab$ &\multicolumn{3}{l}{\hspace{-3mm}$= \Addr \finto \wp(\Store)$} \\
    $\mktab_{\makont} \in \KStore$ &\multicolumn{3}{l}{\hspace{-3mm}$= \ExactContext \finto \wp(\sa{Kont})$} \\
    $\mktab_{\mamkont} \in \MKStore$ &\multicolumn{3}{l}{\hspace{-3mm}$= \MContext \finto \wp(\sa{Kont} \times \sa{MKont})$} \\
    $\makont \in \sa{Kont}$ &\hspace{-3mm}$::= \epsilon \alt \kcons{\mkframe}{\mctx} \alt \mctx$ & $\mvkont \in \VKont$ &\hspace{-3mm}$::= \epsilon \alt \mactx$ \\
    $\mamkont \in \sa{MKont}$ &\hspace{-3mm}$::= \epsilon \alt \mmctx$ & $\maval \in \sa{Value}$ &\hspace{-3mm}$::= \mvkont \alt (\mlam,\menv)$
  \end{tabular}
  \caption{Shift/reset abstract semantic spaces}
  \label{fig:shiftreset-spaces}
\end{figure}
The approximation and flattening happens in $\approximate$:
\begin{equation*}
  \approximate : \MKTab \times \Addr \times \SKont \to \MKTab \times \VKont
\end{equation*}
\begin{align*}
  \approximate(\mmktab,\maddr,\epsilon) &= \mmktab,\epsilon \\
  \approximate(\mmktab,\maddr,\kcons{\mkframe}{\mctx}) &= \mmktab',\kcons{\mkframe}{\mactx} \text{ where } (\mmktab',\mactx) = \approximate(\mmktab,\maddr,\mctx) \\
  \approximate(\mmktab,\maddr,\tpl{\mexpr,\menv,\mastore,\mmktab'}) &= \joinm{\mmktab\sqcup\mmktab'}{\maddr}{\mastore},\kcons{\mkframe}{\tpl{\mexpr,\menv,\maddr}} \\
  \approximate(\mmktab,\maddr,\tpl{\mexpr,\menv,\maddralt}) &= \joinm{\mmktab}{\maddr}{\mmktab(\maddralt)},\kcons{\mkframe}{\tpl{\mexpr,\menv,\maddr}}
\end{align*}
The third case is where continuation closures get flattened together.
The fourth case is when an already approximate continuation is approximated: the approximation is inherited.
Approximating the context and allocating the continuation in the store require two addresses, so we relax the specification of $\alloc$ to allow multiple address allocations in this case.

Each of the four rules of the original shift/reset machine has a corresponding rule that we explain piecemeal.
We will use $\kindastepto$ for steps that do not modify the continuation stores for notational brevity.
We use the above $\approximate$ function in the rule for continuation capture, as modified here.
\begin{equation*}\ev{\sshift{\mvar}{\mexpr},\menv,\mastore,\mmktab,\makont,\mamkont} \kindastepto
  \ev{\mexpr,\menv',\mastore',\mmktab',\epsilon,\mamkont}
\end{equation*}
where
\begin{align*}
  (\maddr,\maddr') &= \alloc(\mastate,\mktab_\makont,\mktab_\mamkont) & \menv' &= \extm{\menv}{\mvar}{\maddr} \\
  (\mmktab',\mvkont) &= \approximate(\mmktab,\maddr',\makont) &
  \mastore' &= \joinm{\mastore}{\maddr}{\mvkont}
\end{align*}

The rule for {\tt reset} stores the continuation and meta-continuation in $\mktab_{\mamkont}$:
\begin{align*}
\ev{\sreset{\mexpr},\menv,\mastore,\mmktab,\makont,\mamkont},\mktab_\makont,\mktab_\mamkont &\stepto
  \ev{\mexpr,\menv,\mastore,\mmktab,\epsilon,\mmctx},\mktab_\makont,\mktab_\mamkont' \\
  \text{where } \mmctx &= \tpl{\mexpr,\menv,\mastore,\mmktab} \\
                \mktab_\mamkont &= \joinm{\mktab_{\mamkont}}{\mmctx}{(\makont,\mamkont)}
\end{align*}

The prompt-popping rule simply dereferences $\mktab_{\mamkont}$:
\begin{align*}
  \co{\epsilon,\mmctx,\maval,\mastore,\mmktab} &\kindastepto \co{\makont,\mamkont,\maval,\mastore,\mmktab} \text{ if } (\makont,\mamkont) \in \mktab_{\mamkont}(\mmctx)
\end{align*}

The continuation installation rule extends $\mktab_{\mamkont}$ at the different context:
\begin{align*}
  \co{\makont,\mamkont,\maval,\mastore,\mmktab},\mktab_\makont,\mktab_\mamkont &\stepto \co{\mvkont,\mmctx,\maval,\mastore,\mmktab},\mktab_\makont,\mktab_\mamkont' \\ 
\text{if } & (\appr{\mvkont},\makont') \in \pop(\mktab_{\makont},\mmktab, \makont) \\
\text{where } \mmctx &= \tpl{\mvkont,\maval,\mastore,\mmktab} \\
              \mktab_\mamkont &= \joinm{\mktab_{\mamkont}}{\mmctx}{(\makont',\mamkont)}
\end{align*}
Again we have a metafunction $\pop$, but this time to interpret approximated continuations:
\begin{align*}
  \pop(\mktab_{\makont}, \mmktab, \makont) &= \popaux(\makont,\emptyset) \\
  \text{where } 
   \popaux(\epsilon, G) &= \emptyset \\
   \popaux(\kcons{\mkframe}{\mctx}, G) &= \set{(\mkframe,\mctx)} \\
   \popaux(\mctx, G) &= \bigcup\limits_{\makont \in G'}(\popaux(\makont, G\cup G')) \\
    \text{where } G' &= \bigcup\limits_{\msctx \in I(\mctx)}{\mktab_{\makont}(\msctx)} \setminus G \\
  I(\msctx) &= \set{\msctx} \\
  I(\tpl{\mexpr,\menv,\maddr}) &=
  \setbuild{\tpl{\mexpr,\menv,\mastore,\mmktab'} \in \dom(\mktab_{\makont})}
           {\mastore \in \mmktab(\maddr),
            \mmktab' \sqsubseteq \mmktab}
\end{align*}
Notice that since we flatten $\mmktab$s together, we need to compare for containment rather than for equality (in $I$).
A variant of this semantics with GC is available in the PLT redex models.

\paragraph{Comparison to CPS transform to remove {\tt shift} and {\tt reset}:}{
We lose precision if we use a CPS transform to compile away {\tt shift} and {\tt reset} forms, because variables are treated less precisely than continuations.
%
%
%
Consider the following program and its CPS transform for comparison:
\begin{small}
\begin{SCodeFlow}\begin{RktBlk}\begin{SingleColumn}\RktPn{(}\RktSym{let*}\mbox{\hphantom{\Scribtexttt{x}}}\RktPn{(}\RktPn{[}\RktSym{id}\mbox{\hphantom{\Scribtexttt{x}}}\RktPn{(}\RktSym{$\lambda$}\mbox{\hphantom{\Scribtexttt{x}}}\RktPn{(}\RktSym{x}\RktPn{)}\mbox{\hphantom{\Scribtexttt{x}}}\RktSym{x}\RktPn{)}\RktPn{]}

\mbox{\hphantom{\Scribtexttt{xxxxxxx}}}\RktPn{[}\RktSym{f}\mbox{\hphantom{\Scribtexttt{x}}}\RktPn{(}\RktSym{$\lambda$}\mbox{\hphantom{\Scribtexttt{x}}}\RktPn{(}\RktSym{y}\RktPn{)}\mbox{\hphantom{\Scribtexttt{x}}}\RktPn{(}\RktSym{shift}\mbox{\hphantom{\Scribtexttt{x}}}\RktSym{k}\mbox{\hphantom{\Scribtexttt{x}}}\RktPn{(}\RktSym{k}\mbox{\hphantom{\Scribtexttt{x}}}\RktPn{(}\RktSym{k}\mbox{\hphantom{\Scribtexttt{x}}}\RktSym{y}\RktPn{)}\RktPn{)}\RktPn{)}\RktPn{)}\RktPn{]}

\mbox{\hphantom{\Scribtexttt{xxxxxxx}}}\RktPn{[}\RktSym{g}\mbox{\hphantom{\Scribtexttt{x}}}\RktPn{(}\RktSym{$\lambda$}\mbox{\hphantom{\Scribtexttt{x}}}\RktPn{(}\RktSym{z}\RktPn{)}\mbox{\hphantom{\Scribtexttt{x}}}\RktPn{(}\RktSym{reset}\mbox{\hphantom{\Scribtexttt{x}}}\RktPn{(}\RktSym{id}\mbox{\hphantom{\Scribtexttt{x}}}\RktPn{(}\RktSym{f}\mbox{\hphantom{\Scribtexttt{x}}}\RktSym{z}\RktPn{)}\RktPn{)}\RktPn{)}\RktPn{)}\RktPn{]}\RktPn{)}

\mbox{\hphantom{\Scribtexttt{xx}}}\RktPn{(}\RktSym{$\le$}\mbox{\hphantom{\Scribtexttt{x}}}\RktPn{(}\RktSym{g}\mbox{\hphantom{\Scribtexttt{x}}}\RktVal{0}\RktPn{)}\mbox{\hphantom{\Scribtexttt{x}}}\RktPn{(}\RktSym{g}\mbox{\hphantom{\Scribtexttt{x}}}\RktVal{1}\RktPn{)}\RktPn{)}\RktPn{)}\end{SingleColumn}\end{RktBlk}\end{SCodeFlow}

\begin{SCodeFlow}\begin{RktBlk}\begin{SingleColumn}\RktPn{(}\RktSym{let*}\mbox{\hphantom{\Scribtexttt{x}}}\RktPn{(}\RktPn{[}\RktSym{id}\mbox{\hphantom{\Scribtexttt{x}}}\RktPn{(}\RktSym{$\lambda$}\mbox{\hphantom{\Scribtexttt{x}}}\RktPn{(}\RktSym{x}\mbox{\hphantom{\Scribtexttt{x}}}\RktSym{k}\RktPn{)}\mbox{\hphantom{\Scribtexttt{x}}}\RktPn{(}\RktSym{k}\mbox{\hphantom{\Scribtexttt{x}}}\RktSym{x}\RktPn{)}\RktPn{)}\RktPn{]}

\mbox{\hphantom{\Scribtexttt{xxxxxxx}}}\RktPn{[}\RktSym{f}\mbox{\hphantom{\Scribtexttt{x}}}\RktPn{(}\RktSym{$\lambda$}\mbox{\hphantom{\Scribtexttt{x}}}\RktPn{(}\RktSym{y}\mbox{\hphantom{\Scribtexttt{x}}}\RktSym{j}\RktPn{)}\mbox{\hphantom{\Scribtexttt{x}}}\RktPn{(}\RktSym{j}\mbox{\hphantom{\Scribtexttt{x}}}\RktPn{(}\RktSym{j}\mbox{\hphantom{\Scribtexttt{x}}}\RktSym{y}\RktPn{)}\RktPn{)}\RktPn{)}\RktPn{]}

\mbox{\hphantom{\Scribtexttt{xxxxxxx}}}\RktPn{[}\RktSym{g}\mbox{\hphantom{\Scribtexttt{x}}}\RktPn{(}\RktSym{$\lambda$}\mbox{\hphantom{\Scribtexttt{x}}}\RktPn{(}\RktSym{z}\mbox{\hphantom{\Scribtexttt{x}}}\RktSym{h}\RktPn{)}
\\\mbox{\hphantom{\Scribtexttt{xxxxxxxxxxx}}}\RktPn{(}\RktSym{h}\mbox{\hphantom{\Scribtexttt{x}}}\RktPn{(}\RktSym{f}\mbox{\hphantom{\Scribtexttt{x}}}\RktSym{z}\mbox{\hphantom{\Scribtexttt{x}}}\RktPn{(}\RktSym{$\lambda$}\mbox{\hphantom{\Scribtexttt{x}}}\RktPn{(}\RktSym{fv}\RktPn{)}
\\\mbox{\hphantom{\Scribtexttt{xxxxxxxxxxxxxxxxxxxx}}}\RktPn{(}\RktSym{id}\mbox{\hphantom{\Scribtexttt{x}}}\RktSym{fv}\mbox{\hphantom{\Scribtexttt{x}}}\RktPn{(}\RktSym{$\lambda$}\mbox{\hphantom{\Scribtexttt{x}}}\RktPn{(}\RktSym{i}\RktPn{)}\mbox{\hphantom{\Scribtexttt{x}}}\RktSym{i}\RktPn{)}\RktPn{)}\RktPn{)}\RktPn{)}\RktPn{)}\RktPn{)}\RktPn{]}\RktPn{)}

\mbox{\hphantom{\Scribtexttt{xx}}}\RktPn{(}\RktSym{g}\mbox{\hphantom{\Scribtexttt{x}}}\RktVal{0}\mbox{\hphantom{\Scribtexttt{x}}}\RktPn{(}\RktSym{$\lambda$}\mbox{\hphantom{\Scribtexttt{x}}}\RktPn{(}\RktSym{g0v}\RktPn{)}\mbox{\hphantom{\Scribtexttt{x}}}\RktPn{(}\RktSym{g}\mbox{\hphantom{\Scribtexttt{x}}}\RktVal{1}\mbox{\hphantom{\Scribtexttt{x}}}\RktPn{(}\RktSym{$\lambda$}\mbox{\hphantom{\Scribtexttt{x}}}\RktPn{(}\RktSym{g1v}\RktPn{)}\mbox{\hphantom{\Scribtexttt{x}}}\RktPn{(}\RktSym{$\le$}\mbox{\hphantom{\Scribtexttt{x}}}\RktSym{g0v}\mbox{\hphantom{\Scribtexttt{x}}}\RktSym{g1v}\RktPn{)}\RktPn{)}\RktPn{)}\RktPn{)}\RktPn{)}\RktPn{)}\end{SingleColumn}\end{RktBlk}\end{SCodeFlow}
\end{small}
The $\CESKKstart$ machine with a monovariant allocation strategy will predict the CPS'd version returns true or false.
In analysis literature, ``monovariant'' means variables get one address, namely themselves.
Our specialized analysis for delimited control will predict the non-CPS'd version returns true.}

\iftr{
\subsection{Correctness}
We impose an order on values since stored continuations are more approximate in the analysis than in $\SR$:
\begin{mathpar}
  \inferrule{ }{\mval \sqsubseteq_{\mktab,\mmktab} \mval} \quad
  \inferrule{\mkont \sqsubseteq \unroll{\mktab,\mmktab}{\mvkont}}
            {\vcomp{\mkont} \sqsubseteq_{\mktab,\mmktab} \mvkont} \quad
  \inferrule{\forall \mval\in\mstore(\maddr).
             \exists \maval\in\mastore(\maddr).
             \mval \sqsubseteq_{\mktab,\mmktab} \maval}
            {\mstore \sqsubseteq_{\mktab,\mmktab} \mastore} \\
  \inferrule{\mkont \sqsubseteq \unroll{\mktab_{\makont},\mmktab}{\makont} \\
             \mmkont \sqsubseteq \unrollC{\mktab_{\makont},\mktab_{\mamkont},\mmktab}{\mamkont} \\
             \mstore \sqsubseteq_{\mktab_{\makont},\mmktab} \mastore}
            {\ev{\mexpr,\menv,\mstore,\mkont,\mmkont} \sqsubseteq
             \ev{\mexpr,\menv,\mastore, \mmktab,\makont,\mamkont}, \mktab_{\makont}, \mktab_{\mamkont}} \\
  \inferrule{\mval \sqsubseteq_{\mktab_{\makont},\mmktab} \maval \\
             \mkont \sqsubseteq \unroll{\mktab_{\makont},\mmktab}{\makont} \\
             \mmkont \sqsubseteq \unrollC{\mktab_{\makont},\mktab_{\mamkont},\mmktab}{\mamkont} \\
             \mstore \sqsubseteq_{\mktab_{\makont},\mmktab} \mastore}
            {\co{\mkont,\mmkont,\mval,\mstore} \sqsubseteq
             \co{\makont,\mamkont,\maval,\mastore, \mmktab}, \mktab_{\makont}, \mktab_{\mamkont}}
\end{mathpar}
Unrolling differs from the previous sections because the values in frames can be approximate.
Thus, instead of expecting the exact continuation to be in the unrolling, we have a judgment that an unrolling approximates a given continuation in \autoref{fig:cont-order} (note we reuse $I$ from $\popaux$'s definition).

\begin{figure}
  \centering
  \begin{mathpar}
    \inferrule{ }{\appl{\mexpr,\menv} \sqsubseteq_{\mktab,\mmktab}
      \appl{\mexpr,\menv}} \quad \inferrule{\mval
      \sqsubseteq_{\mktab,\mmktab}{\maval}}
    {\appr{\mval} \sqsubseteq_{\mktab,\mmktab} \appr{\maval}} \\
    \inferrule{ }{\epsilon \sqsubseteq
      \unroll{\mktab,\mmktab}{\epsilon}} \quad
    \inferrule{\mkframe \sqsubseteq_{\mktab,\mmktab} \makframe \\
      \mkont \sqsubseteq \unroll{\mktab,\mmktab}{\mctx}}
    {\kcons{\mkframe}{\mkont} \sqsubseteq
      \unroll{\mktab,\mmktab}{\kcons{\makframe}{\mctx}}}
    \\
    \inferrule{\makont \in \mktab(\msctx) \quad
      \mkont \sqsubseteq \unroll{\mktab,\mmktab}{\makont}} {\mkont
      \sqsubseteq \unroll{\mktab,\mmktab}{\msctx}}
    \quad
    \inferrule{\msctx \in I(\mktab,\mmktab,\mactx) \quad
      \mkont \sqsubseteq \unroll{\mktab,\mmktab}{\msctx}} {\mkont
      \sqsubseteq \unroll{\mktab,\mmktab}{\mactx}}
    \\
    \inferrule{ }
              {\epsilon \sqsubseteq \unrollC{\mktab_{\makont},\mktab_{\mamkont},\mmktab}{\epsilon}}
    \\
    \inferrule{(\makont,\mamkont) \in \mktab_{\mamkont}(\mmctx) \\
               \mkont \sqsubseteq \unroll{\mktab_{\makont},\mmktab}{\makont} \\
               \mmkont \sqsubseteq \unrollC{\mktab_{\makont},\mktab_{\mamkont},\mmktab}{\mamkont}}
              {\mkapp{\mkont}{\mmkont} \sqsubseteq \unrollC{\mktab_{\makont},\mktab_{\mamkont},\mmktab}{\mmctx}}
  \end{mathpar}
  
  \caption{Order on (meta-)continuations}
\label{fig:cont-order}
\end{figure}
\begin{theorem}[Soundness]
  If $\somestate \stepto_{\SR} \nextstate$, and $\somestate \sqsubseteq \someotherstate$ then there is $\nextotherstate$ such that $\someotherstate \stepto_{\SRSChKKt} \nextotherstate$ and
$\nextstate \sqsubseteq \nextotherstate$.
\end{theorem}

\paragraph{Freshness implies completeness}
The high level proof idea is that fresh allocation separates evaluation into a sequence of bounded length paths that have the same store, but the store only grows and distinguishes contexts such that each continuation and metacontinuation have a unique unrolling.
It is an open question whether the addition of garbage collection preserves completeness.
Each context with the same store will have different expressions in them since expressions can only get smaller until a function call, at which point the store grows.
This forms an order on contexts: smaller store means smaller context, and same store but smaller expression (indeed a subexpression) means a smaller context.
Every entry in each enviroment ($\mastore,\mmktab,\mktab_\makont,\mktab_\mamkont$) will map to a unique element, and the continuation stores will have no circular references (the context in the tail of a continuation is strictly smaller than the context that maps to the continuation).
There can only be one context that $I$ maps to for approximate contexts because of the property of stores in contexts.

We distill these intuitions into an invariant about states that we will then use to prove completeness.
\begin{mathpar}
  \inferrule{\forall \maddr\in\dom(\mastore).\exists \maval. \mastore(\maddr)=\set{\maval}\wedge\maval\preceq_\mmktab\mktab_\makont \\
    \forall \maddr\in\dom(\mmktab).\exists\mastore'.\mmktab(\maddr) =\set{\mastore'}\wedge\mastore'\in\pi_3(\dom(\mktab_\makont)) \\
    \forall \msctx\in\dom(\mktab_\makont).\exists\makont. \mktab_\makont(\msctx) = \set{\makont}\wedge \makont \sqsubset_\mmktab^{\mktab_\makont} \msctx\\
    \forall \mmctx\in\dom(\mktab_\mamkont).\exists\mamkont.\mktab_\mamkont(\mmctx) = \set{\mamkont}\wedge \mamkont \sqsubset \mmctx \\
}{\inv^*(\mastore, \mmktab, \mktab_{\makont}, \mktab_{\mamkont})}
 \\
\inferrule{\inv^*(\mastore,\mmktab,\mktab_\makont,\mktab_\mamkont) \\
           \tpl{\mexpr,\menv,\mastore,\mmktab} \sqsubset \dom(\mktab_\makont) \cup \dom(\mktab_\mamkont) \\
           (\exists \tpl{\mexpr_c,\menv,\mastore,\mmktab} \in \dom(\mktab_\makont)) \implies \mexpr \in \mathit{subexpressions}(\mexpr_c) \\
           \makont \preceq_\mmktab \mktab_\makont \\
           \mamkont \preceq \mktab_\mamkont}
          {\inv_\fresh(\ev{\mexpr,\menv,\mastore,\mmktab,\makont,\mamkont},\mktab_\makont,\mktab_\mamkont)} \\
\inferrule{\inv^*(\mastore,\mmktab,\mktab_\makont,\mktab_\mamkont) \\
           \maval \preceq_\mmktab \mktab_\makont \\
           \makont \preceq_\mmktab \mktab_\makont \\
           \mamkont \preceq \mktab_\mamkont}
          {\inv_\fresh(\co{\makont,\mamkont,\maval,\mastore,\mmktab},\mktab_\makont,\mktab_\mamkont)}
\end{mathpar}
Where the order $\preceq$ states that any contexts in the (meta-)continuation are mapped in the given table.
\begin{mathpar}
  \inferrule{ }{(\mlam,\menv) \preceq_\mmktab \mktab_\makont} \quad
  \inferrule{ }{\epsilon \preceq_\mmktab \mktab_\makont} \quad
  \inferrule{ }{\epsilon \preceq \mktab_\mamkont} \quad
  \inferrule{\msctx \in \dom(\mktab_\makont)}{\msctx\preceq_\mmktab \mktab_\makont} \quad
  \inferrule{\mmctx \in \dom(\mktab_\mamkont)}{\mmctx \preceq \mktab_\mamkont}\\
  \inferrule{\exists\mastore.\mmktab(\maddr)=\set{\mastore} \\
             \exists!\mmktab'.\tpl{\mexpr,\menv,\mastore,\mmktab'} \in\dom(\mktab_\makont)\wedge\mmktab' \sqsubseteq \mmktab
           }
            {\tpl{\mexpr,\menv,\maddr} \preceq_\mmktab \mktab_\makont}
\end{mathpar}
And the order $\sqsubset$ states that the contexts in the (meta-)continuation are strictly smaller than the given context.
\begin{mathpar}
  \inferrule{ }{\epsilon \sqsubset_\mmktab^{\mktab_{\makont}}} \quad \inferrule{ }{\epsilon \sqsubset \mmctx} \quad \inferrule{\mctx \sqsubset_\mmktab^{\mktab_\makont} \msctx}{\kcons{\mkframe}{\mctx} \sqsubset_\mmktab^{\mktab_\makont} \msctx} \\
  \inferrule{\mexpr' \in \mathit{subexpressions}(\mexpr)}{\tpl{\mexpr',\menv,\mastore,\mmktab} \sqsubset_\mmktab^{\mktab_\makont} \tpl{\mexpr,\menv,\mastore,\mmktab}} \quad
  \inferrule{\mexpr' \in \mathit{subexpressions}(\mexpr)}{\tpl{\mexpr',\menv,\mastore,\mmktab} \sqsubset \tpl{\mexpr,\menv,\mastore,\mmktab}} \\
  \inferrule{\dom(\mastore) \sqsubset \dom(\mastore')}{\tpl{\_,\_,\mastore,\_} \sqsubset_\mmktab^{\mktab_\makont} \tpl{\_,\_,\mastore',\_}} \quad
  \inferrule{\dom(\mastore) \sqsubset \dom(\mastore')}{\tpl{\_,\_,\mastore,\_} \sqsubset \tpl{\_,\_,\mastore',\_}} \\
  \inferrule{\forall \msctx' \in I(\mktab_\makont,\mmktab,\mactx) \\ \msctx' \sqsubset_\mmktab^{\mktab_\mkont} \msctx}{\mactx \sqsubset_\mmktab^{\mktab_\makont} \msctx}
\end{mathpar}

\begin{lemma}[Freshness invariant]
  If $\alloc$ produces fresh addresses, $\inv_{\fresh}(\mastate,\mktab_\makont,\mktab_\mamkont)$ and
$\mastate,\mktab_{\makont},\mktab_{\mamkont} \stepto
\mastate',\mktab'_{\makont},\mktab'_{\mamkont}$ then
$\inv_{\fresh}(\mastate',\mktab'_{\makont},\mktab'_{\mamkont})$.
\end{lemma}
\begin{proof}
  By case analysis on the step.
\end{proof}
\begin{theorem}[Complete for fresh allocation]
  If $\alloc$ produces fresh addresses then the resulting semantics is complete with respect to states satisfying the invariant.
\end{theorem}
\begin{proof}[Proof sketch]
  By case analysis and use of the invariant to exploit the fact the unrollings are unique and the singleton codomains pigeon-hole the possible steps to only concrete ones.
\end{proof}
}

%% file: memo.tex
All the semantics of previous sections have a performance weakness that many analyses share: unnecessary propagation.
Consider two portions of a program that do not affect one another's behavior.
Both can change the store, and the semantics will be unaware that the changes will not interfere with the other's execution.
The more possible stores there are in execution, the more possible contexts in which a function will be evaluated.
Multiple independent portions of a program may be reused with the same arguments and store contents they depend on, but changes to irrelevant parts of the store lead to redundant computation.
The idea of skipping from a change past several otherwise unchanged states to uses of the change is called ``sparseness'' in the literature~\citep{dvanhorn:Reif1977Symbolic,dvanhorn:Wegman1991Constant,dvanhorn:Oh2012Design}.
Memoization is a specialized instance of sparseness; the base stack may change, but the evaluation of the function does not, so given an already computed result we can jump straight to the answer.
We use the vocabulary of ``relevance'' and ``irrelevance'' so that future work can adopt the ideas of sparseness to reuse contexts in more ways.

Recall the core notion of irrelevance: if we have seen the results of a computation before from a different context, we can reuse them.
The semantic counterpart to this idea is a memo table that we extend when popping and appeal to when about to push.
This simple idea works well with a deterministic semantics, but the non-determinism of abstraction requires care.
In particular, memo table entries can end up mapping to multiple results, but not all results will be found at the same time.
Note the memo table space:
\begin{align*}
  \mmemo \in \Memo &= \Context \finto \wp(\Relevant) \\
  \Relevant &::= \tpl{\mexpr,\menv,\mstore}
\end{align*}
There are a few ways to deal with multiple results:
\begin{enumerate}
\item{rerun the analysis with the last memo table until the table doesn't change (expensive),}
\item{short-circuit to the answer but also continue evaluating anyway (negates most benefit of short-circuiting), or}
\item{use a frontier-based semantics like in \autoref{sec:eng-frontier} with global $\mktab$ and $\mmemo$, taking care to
    \begin{enumerate}
    \item{at memo-use time, still extend $\mktab$ so later memo table extensions will ``flow'' to previous memo table uses, and}
    \item{when $\mktab$ and $\mmemo$ are extended at the same context at the same time, also create states that act like the $\mmemo$ extension point also returned to the new continuations stored in $\mktab$.}
    \end{enumerate}}
\end{enumerate}

We will only discuss the final approach.
The same result can be achieved with a one-state-at-a-time frontier semantics, but we believe this is cleaner and more parallelizable.
Its second sub-point we will call the ``push/pop rendezvous.''
The rendezvous is necessary because there may be no later push or pop steps that would regularly appeal to either (then extended) table at the same context.
The frontier-based semantics then makes sure these pushes and pops find each other to continue on evaluating.
In pushdown and nested word automata literature, the push to pop short-circuiting step is called a ``summary edge'' or with respect to the entire technique, ``summarization.''
We find the memoization analogy appeals to programmers' and semanticists' operational intuitions.
A second concern for using memo tables is soundness.
Without the completeness property of the semantics, memoized results in, \eg{}, an inexactly GC'd machine, can have dangling addresses since the possible stacks may have grown to include addresses that were previously garbage.
These addresses would not be garbage at first, since they must be mapped in the store for the contexts to coincide, but during the function evaluation the addresses can become garbage.
If they are supposed to then be live, and are used (presumably they are reallocated post-collection), the analysis will miss paths it must explore for soundness.

\iftr{
Context-irrelevance is a property of the semantics \emph{without} continuation stores, so there is an additional invariant to that of \autoref{sec:pushdown} for the semantics with $\mktab$ and $\mmemo$: $\mmemo$ respects context irrelevance.
\begin{mathpar}
  \inferrule{\dom(\mmemo) \subseteq \dom(\mktab) \\
             \forall \mctx\equiv\tpl{\mexpr_c,\menv_c,\mstore_c} \in \dom(\mmemo),
                     \tpl{\mexpr_r,\menv_r,\mstore_r} \in \mmemo(\mctx), \\
                     \makont\in\mktab(\mctx),
                     \mkont\in\unroll{\mktab}{\makont}. \\
              \exists\mtrace\equiv\tpl{\mexpr_c,\menv_c,\mstore_c,\mkont} \stepto_{\CESKt}^* \tpl{\mexpr_r,\menv_r,\mstore_r,\mkont}. \hastail(\mtrace,\mkont)}
            {\inv_M(\mktab,\mmemo)}
\end{mathpar}
Inexact GC does \emph{not} respect context irrelevance for the same reasons it is not complete: some states are spurious, meaning some memo table entries will be spurious, and the expected path in the invariant will not exist.
The reason we use unrolled continuations instead of simply $\epsilon$ for this (balanced) path is precisely for stack inspection reasons.
}

 \begin{figure}
   \begin{center}
     $\mastate,\mktab,\mmemo \stepto
     \mastate',\mktab',\mmemo'$
     \begin{tabular}{r|l}
       \hline\vspace{-3mm}\\
       $\tpl{\sapp{\mexpri0}{\mexpri1},\menv,\mstore,\makont},\mktab,\mmemo$
       &
       $\tpl{\mexpri0,\menv,\mstore,\kcons{\appl{\mexpri1,\menv}}{\mctx}},\mktab,\mmemo$ \\
       & \quad if $\mctx \notin\dom(\mmemo)$, or \\
       &
       $\tpl{\mexpr',\menv',\mstore',\makont},\mktab',\mmemo$ \\
       & \quad if $\tpl{\mexpr',\menv',\mstore'} \in \mmemo(\mctx)$ \\
       where & $\mctx = \tpl{\sapp{\mexpri0}{\mexpri1},\menv,\mstore}$ \\
       & $\mktab' = \joinm{\mktab}{\mctx}{\makont}$
       \\
       $\tpl{\mval,\mstore,\kcons{\appr{\slam{\mvar}{\mexpr},\menv}}{\mctx}},\mktab,\mmemo$
       &
       $\tpl{\mexpr,\menv',\mstore',\makont},\mktab,\mmemo'$ if $\makont \in \mktab(\mctx)$ \\
       where & $\menv' = \extm{\menv}{\mvar}{\maddr}$ \\
       & $\mstore' = \joinm{\mstore}{\maddr}{\mval}$ \\
       & $\mmemo' = \joinm{\mmemo}{\mctx}{\tpl{\mexpr,\menv',\mstore'}}$
     \end{tabular}
   \end{center}
   \caption{Important memoization rules}
   \label{fig:memo}
 \end{figure}

The rules in \autoref{fig:memo} are the importantly changed rules from \autoref{sec:pushdown} that short-circuit to memoized results.
The technique looks more like memoization with a $\CESIKKstart$ machine, since the memoization points are truly at function call and return boundaries.
The $\pop$ function would need to also update $\mmemo$ if it dereferences through a context, but otherwise the semantics are updated \emph{mutatis mutandis}.

\begin{equation*}
  {\mathcal F}_{\mexpr}(S,R,F,\mktab,\mmemo) = (S \cup F, R \cup R', F'\setminus S, \mktab', \mmemo')
\end{equation*}
where

\begin{tabular}{rlrlrl}
  $I$ &
  \multicolumn{5}{l}{
    \hspace{-3mm}$=\bigcup\limits_{\mstate \in
      F}{\setbuild{(\tpl{\mstate,\mstate'}, \mktab',\mmemo')}{\mstate,\mktab,\mmemo
        \stepto \mstate',\mktab',\mmemo'}}$}
\\
   $R'$ &\hspace{-3mm}$= \pi_0 I$ & $\mktab'$ & \hspace{-3mm}$= \bigsqcup\pi_1 I$ & $\mmemo'$ & \hspace{-3mm}$= \bigsqcup\pi_2 I$ \\
   $\Delta\mktab$ &\hspace{-3mm}$= \mktab'\setminus\mktab$ & $\Delta\mmemo$ & \hspace{-3mm}$= \mmemo'\setminus\mmemo$ & & \\
   $F'$ &
   \multicolumn{5}{l}{
     \hspace{-3mm}$= \pi_1 R' \cup \{{\tpl{\mexpr,\menv,\mstore,\makont}} :
     {\mctx \in \dom(\Delta\mktab)\cap\dom(\Delta\mmemo).}$}
   \\ &\multicolumn{5}{l}{\hspace{-3mm}$\phantom{= \pi_1 R' \cup \{} \makont \in \Delta\mktab(\mctx),
       \tpl{\mexpr,\menv,\mstore} \in \Delta\mmemo(\mctx)\}$}
 \end{tabular}

The $\pi_i$ notation is for projecting out pairs, lifted over sets.
This worklist algorithm describes unambiguously what is meant by ``rendezvous.''
After stepping each state in the frontier, the differences to the $\mktab$ and $\mmemo$ tables are collected and then combined in $F'$ as calling contexts' continuations matched with their memoized results.

\iftr{
\begin{theorem}[Correctness]
Same property is the same as in \autoref{thm:global-pushdown}, where $\reify$ ignores the $\mmemo$ component.
\end{theorem}
The proof appeals to the invariant on $\mmemo$ whose proof involves an additional argument for the short-circuiting step that reconstructs the path from a memoized result using both context irrelevance and the table invariants.
}

%% file: related.tex
The immediately related work is that of PDCFA \citep{dvanhorn:Earl2010Pushdown, dvanhorn:Earl2012Introspective}, CFA2~\citep{dvanhorn:Vardoulakis2011CFA2, dvanhorn:Vardoulakis2011Pushdown}, and AAM~\citep{dvanhorn:VanHorn2010Abstracting}.
The stack frames in CFA2 that boost precision are an orthogonal feature that fit into our model as an \emph{irrelevant} component along with the stack, which we did not cover in detail due to space constraints.
The version of CFA2 that handles \rackett{call/cc} does not handle composable control, is dependent on a restricted CPS representation, and has untunable precision for first-class continuations.
Our semantics adapts to \rackett{call/cc} by removing the meta-continuation operations, and thus this work supersedes theirs; the machinery is in fact a strict generalization.
The extended version of PDCFA that inspects the stack to do garbage collection~\citep{dvanhorn:Earl2012Introspective} also fits into our model (\autoref{sec:gc}'s $\hat\Gamma$).
We suspect the more ``semantic'' garbage collection from \citet{mc-via-gamma} can be easily adapted to the pushdown setting.

We did additional work to improve the performance of the AAM approach in \citet{dvanhorn:Johnson2013Optimizing} that can almost entirely be imported for the work in this paper.
The technique that does not apply is ``store counting'' for lean representation of the store component of states when there is a global abstract store, an assumption that does not hold for garbage collection.
The implementation\footnote{\url{http://github.com/dvanhorn/oaam}} has pushdown modules that use the ideas in this paper.

\paragraph{Stack inspection}
Stack inspecting flow analyses also exist, but operate on pre-constructed regular control-flow graphs~\citep{ianjohnson:bartoletti2004stack}, so the CFGs cannot be trimmed due to the extra information at construction time, leading to less precision.
Backward analyses for stack inspection also exist, with the same prerequisite~\citep{ianjohnson:DBLP:journals/sigplan/Chang06}.

\paragraph{Pushdown models and memoization}
The idea of relating pushdown automata with memoization is not new.
In 1971, Stephen Cook~\citep{DBLP:conf/ifip/Cook71} devised a transformation to simulate 2-way (on a fixed input) \emph{deterministic} pushdown automata in time linear in the size of the input, that uses the same ``context irrelevance'' idea to skip from state $q$ seen before to a corresponding first state that pops the stack below where $q$ started (called a \emph{terminator} state).
Six years later, Neil D. Jones~\citep{Jones:1977:NLT} simplified the transformation instead to \emph{interpret} a stack machine program to work \emph{on-the-fly} still on a deterministic machine, but with the same idea of using memo tables to remember corresponding terminator states.
Thirty-six years after that, at David Schmidt's Festschrift, Robert Gl\"uck extended the technique to 2-way \emph{non-deterministic} pushdown automata, and showed that the technique can be used to recognize context-free languages in the standard ${\mathcal O}(n^3)$ time~\citep{DBLP:journals/corr/Gluck13}.
Gl\"uck's technique (as yet, correctness unproven) uses the meta-language's stack with a deeply recursive interpretation function to preclude the use of a frontier and something akin to $\mktab$\footnote{See \texttt{gluck.rkt} in online materials for Gl\"uck's style}.
By exploring the state space \emph{depth-first}, the interpreter can find all the different terminators a state can reach one-by-one by destructively updating the memo table with the ``latest'' terminator found.
The trade-offs with this technique are that it does not obviously scale to first-class control, and the search can overflow the stack when interpreting moderate-sized programs.
A minor disadvantage is that it is also not a fair evaluation strategy when allocation is unbounded.
The technique can nevertheless be a viable alternative for languages with simple control-flow mechanisms.
It has close similarities to ``Big-CFA2'' in Vardoulakis' dissertation~\citep{vardoulakis-diss12}.
\paragraph{Analysis of pushdown automata}
Pushdown models have existed in the first-order static analysis literature~\citep[Chapter 7]{local:muchnick:jones:flow-analysis:1981}\citep{dvanhorn:Reps1995Precise}, and the first-order model checking literature \citep{dvanhorn:Bouajjani1997Reachability}, for some time.
These methods already assume a pushdown model as input, and constructing a model from a first-order program is trivial.
In the setting of higher-order functions and first-class control, model construction is an additional problem -- the one we solve here.
Additionally, the algorithms employed in these works expect a complete description of the model up front, rather than work with a modified \texttt{step} function (also called \texttt{post}), such as in ``on-the-fly'' model-checking algorithms for finite state systems~\citep{dvanhorn:Schwoon2005Note}.
\paragraph{Derivation from abstract machines}
The trend of deriving static analyses from abstract machines does not stop at flow analyses.
The model-checking community showed how to check temporal logic queries for collapsible pushdown automata (CPDA), or equivalently, higher-order recursion schemes, by deriving the checking algorithm from the Krivine machine~\citep{dvanhorn:Salvati2011Krivine}.
The expressiveness of CPDAs outweighs that of PDAs, but it is unclear how to adapt higher-order recursion schemes (HORS) to model arbitrary programming language features.
The method is strongly tied to the simply-typed call-by-name lambda calculus and depends on finite sized base-types.
%